\documentclass[10p, journal]{IEEEtran}
\renewcommand{\baselinestretch}{0.98}
\newcommand{\squeezeup}{\vspace{-2.5mm}}
\usepackage{lipsum}
\usepackage{etoolbox}
\usepackage{cite}
\usepackage{amsmath,amssymb,amsfonts}
\usepackage{graphicx}
\usepackage{textcomp}
\usepackage[printonlyused,nohyperlinks,nolist]{acronym}
\usepackage{breqn}
\usepackage{ifpdf}
\usepackage{manfnt}
\usepackage[mathscr]{eucal}
\usepackage{graphicx}
\usepackage{epstopdf}
\usepackage{acronym}
\usepackage{tabularx}
\usepackage{tabulary}
\graphicspath{ {./fig/} }
\usepackage{algorithmicx}
\usepackage[ruled,vlined,commentsnumbered]{algorithm2e}
\usepackage{algpseudocode}

\usepackage[center]{caption}

\usepackage{makeidx}
\makeindex

\usepackage{graphicx,color}
\graphicspath{ {./fig/} }

\usepackage{longtable}
\usepackage{afterpage,lipsum}

\usepackage{multirow} 

\usepackage{arydshln,leftidx,mathtools}%
\usepackage{enumerate}
\usepackage{enumitem}
\usepackage[normalem]{ulem}

\usepackage{courier}
\usepackage{amssymb}
\usepackage{amsfonts}
\usepackage{amsmath}
\usepackage{subcaption}
\usepackage{color}
\usepackage{colortbl}
\usepackage{amsthm}

\newtheorem{proposition}{Proposition}

\newtheorem{Def}{Definition}

\newcommand{\set}[1]{{\mathcal{#1}}}

\newcommand{\setE}{\set{E}}

\newcommand{\f}{\boldsymbol{f}}
\newcommand{\w}{\boldsymbol{w}}
\newcommand{\x}{\boldsymbol{x}}
\newcommand{\di}{\boldsymbol{d}}
\newcommand{\shahat}{\hat{\boldsymbol{s}}}
\newcommand{\alphas}{\boldsymbol{\alpha}}
\newcommand{\E}{\mathbb{E}}
\newcommand{\av}{\boldsymbol{a}}
\newcommand{\y}{\boldsymbol{y}}
\newcommand{\Id}{\textup{Id}}
\newcommand{\Tr}{\textup{Tr}}

\newcommand{\Rmnum}[1]{\uppercase\expandafter{\romannumeral #1}}

\newcommand{\setX}{{\mathcal X}}
\newcommand{\setN}{{\mathcal N}}

\newcommand{\setC}{{\mathcal C}}
\newcommand{\setF}{{\mathcal F}}

\newcommand{\setQ}{{\mathcal Q}}
\newcommand{\setY}{{\mathcal Y}}

\newcommand{\limit}{\textup{lim}}
\newcommand{\infi}{\textup{inf}}

\newcommand{\signal}[1]{{\boldsymbol{#1}}}

\newcommand{\real}{{\mathbb R}}

\newtheorem{fact}{Fact}

\theoremstyle{newremark}

\theoremstyle{newremark}
\newtheorem{assumption}{Assumption}
\newtheorem{lemma}{Lemma}
\newtheorem{theorem}{Theorem}
\newcommand{\Natural}{{\mathbb N}}

\begin{document}

\title{Hybrid Model and Data Driven Algorithm for Online Learning of Any-to-Any Path Loss Maps}

\author{M. A. Gutierrez-Estevez,
		Martin Kasparick, Renato L. G. Cavalvante, and 
		S\l awomir Sta\'nczak

\thanks{M. A. Gutierrez-Estevez and S\l awomir Sta\'nczak are with the Network Information Theory, Technical University of Berlin, Germany. Email: m.gutierrezestevez@campus.tu-berlin.de,  slawomir.stanczak@tu-berlin.de.}	
\thanks{Martin Kasparick, and Renato L. G. Cavalvante are with the department of Wireless Communications and Networks of the Fraunhofer Heinrich-Herz Institute in Berlin, Germany. Email: \{martin.kasparick, renato.cavalcante\}@hhi.fraunhofer.de.}

}

\maketitle

\begin{abstract}
Learning any-to-any (A2A) path loss maps, where the objective is the reconstruction of path loss between any two given points in a map, might be a key enabler for many applications that rely on device-to-device (D2D) communication. Such applications include machine-type communications (MTC) or vehicle-to-vehicle (V2V) communications. Current approaches for learning A2A maps are either model-based methods, or pure data-driven methods.  Model-based methods have the advantage that they can generate reliable estimations with low computational complexity, but they cannot exploit information coming from data. Pure data-driven methods can achieve good performance without assuming any physical model, but their complexity and their lack of robustness is not acceptable for many applications. In this paper, we propose a novel hybrid model and data-driven approach that fuses information obtained from datasets and models in an online fashion. To that end, we leverage the framework of stochastic learning to deal with the sequential arrival of samples and propose an online algorithm that alternatively and sequentially minimizes the original non-convex problem. A proof of convergence is presented, along with experiments based firstly on synthetic data, and secondly on a more realistic dataset for V2X, with both experiments showing promising results.

\end{abstract}

\squeezeup
\section{Introduction}



Many applications in wireless networks can benefit from information related to the spatial distribution of path loss. Among them, applications involving peer-to-peer communication are the most challenging ones because of fast increase of communication links when the number of nodes grows. Such applications include sensor networks, \ac{MTC} or \ac{V2V} communications. As an example consider a platoon of vehicles that have to constantly exchange information about their position, acceleration, and so on. If the path loss between any two vehicles along the route was known in advance, this information would give the vehicles enough time to adapt their distance accordingly and save a considerable amount of fuel  \cite{jornod2019packet}. Other benefits include reliability of communications and safety. 

\ac{A2A} maps describe the spatial distribution of radio signals between any two given locations of a map, which makes them very suited for those applications. But the challenge is to cope with the rapid increase of complexity when the map size increases, while keeping high prediction accuracy. 



Some approaches for radio maps estimation are pure data-driven methods \cite{schaufele2019tensor,kasparick2016kernel,chouvardas2016method} in the sense that no physical model for the propagation of radio signals is considered, but instead they exploit the expected spatial correlation of the channel characteristics. Other approaches either rely on fixed mathematical models to describe the propagation of radio signals \cite{lee2017channel,hamilton2013propagation,konak2011predicting}, or they attempt at learning such models without context information \cite{romero2016blind,gutierrez2018nonparametric}. Model-based methods have the advantage that they can generate reliable estimations with low computational complexity and little to no side information. However, they are rigid in the sense that they cannot exploit information coming from data to adapt to the environment and to reduce model uncertainty. On the other hand, pure data-driven methods can achieve good performance without assuming any physical model, but their complexity and their lack of robustness against changes in the environment (e.g., underlying distribution of the data) is not acceptable for many applications.


Against this background, we introduce a novel hybrid data- and model-driven approach with the intention of extracting the best of both worlds: we start with the notion that a mathematical model can coarsely represent the physical world, but we endow our method with the flexibility to modify the original model based on the acquired measurements. Further, our method is online because both the physical characteristics of the environment may vary (disposition of buildings, environmental conditions like rain or fog, etc...), and also because an online method can deal with the high complexity of the problem for large scenarios.



\subsection{Prior Art}

The learning of radio maps has been a major topic of interest both in academia and the industry for years \cite{schaufele2019tensor,kasparick2016kernel,gutierrez2014spatial,konak2011predicting,dall2011channel,lee2017channel,hamilton2013propagation,wilson2010radio,romero2016blind,gutierrez2018nonparametric,chouvardas2016method}. In recent years, the framework of \ac{TPT} has gained a great deal of attention as a model that characterizes the long-term shadowing of links caused by objects such as buildings or trees\cite{agrawal2009correlated,konak2011predicting,patwari2008effects}, and in turn this shadowing is used as a proxy to characterize the path loss. In \ac{TPT}, a \ac{SLF} captures the absorption generated by objects in a field, while a window function models the influence of each location on the attenuation that every link experiences \cite{agrawal2009correlated}. The shadowing is then modeled as the weighted integral of the \ac{SLF} across the field. 

Previous studies \cite{gudmundson1991correlation,agrawal2009correlated} have investigated statistical and correlation properties for shadow fading in different networking scenarios. One of the main challenges related to the statistical modeling of shadow fading lies in the characterization of its spatio-temporal correlation properties. The \ac{SLF} in \cite{agrawal2009correlated} is assumed to be a zero-mean Gaussian random field, and consequently, the shadowing loss experienced on arbitrary links can also be seen as a Gaussian random field. The treatment of shadow fading as a Gaussian random field has led several authors \cite{konak2011predicting,kim2010collaborative,gutierrez2014spatial} to use Kriging interpolation for the estimation of coverage maps. In \cite{kim2010collaborative,dall2011channel}, a state-space extension of the general path loss model is adopted in order to track coverage maps using the Kriged Kalman filter.


A different approach exploits the concept of the Fresnel zone \cite{hamilton2013propagation,patwari2008effects,wilson2010radio,wilson2010see,lee2017channel} to create a model that represents the propagation of wireless signals. In particular, the authors in \cite{hamilton2013propagation} propose different models taking into account the locations of transmitter and receiver, and a weight is assigned to each location contained in the Fresnel zone representing the impact of each location in the signal propagation. These models are then used for different applications such as coverage maps generation \cite{braham2014coverage}, scene reconstruction \cite{wilson2010see}, or path loss estimation \cite{lee2017channel}. In \cite{lee2017channel}, the \ac{SLF} is modeled as the sum of a low rank matrix, which is potentially corrupted by sparse outliers, and a sparse matrix. The motivation for this assumption is that the regular placement of walls and buildings in urban scenarios renders the scene inherently low-rank, while sparse outliers can pick up artifacts that do not conform to the low-rank model. The problem becomes an instance of the compressive principle component pursuit (CPCP) approach, and the authors propose an iterative algorithm to reconstruct the \ac{SLF}. 


In order to overcome the limitations of a mathematical model, the authors in \cite{romero2016blind} propose an algorithm that learns both the \ac{SLF} and the window function in a blind manner, i.e. no model is assumed and both the \ac{SLF} and the window function are learned in an alternating fashion. In \cite{gutierrez2018nonparametric}, this blind approach is further improved by capitalizing on the fact that both structures are assumed to be block-sparse. A problem with elastic net regularization and multi-kernels is formulated, and an algorithm based on the \ac{ADMM} is used to obtain a solution. Both contributions in \cite{romero2016blind,gutierrez2018nonparametric} have the critical limitation of being batch algorithms, which poses a tremendous hurdle for real-world applications because i) in both approaches the problem complexity and the number of variables that have to be stored in memory increases cubicly with the number of pixels in the map, and ii) because they cannot cope with a changing environment over time due to e.g. different traffic profiles or change in the underlying map. 
In the seeding publication of this work \cite{letter}, we overcome these limitations by proposing an online algorithm which, upon arrival of new measurements, obtains new estimates of both the \ac{SLF} and the model. To do this, the online algorithm implements a ``descent" version of the \ac{gAM} \cite{jain2017non}, i.e. we take only one step at a time towards a new estimate of the \ac{SLF} with the last estimate of the model fixed, and then another step for the model with the new updated \ac{SLF} fixed, iteratively until a stopping criterion is met.

\subsection{Contributions}

In the following we enumerate the contributions of this work:
\begin{enumerate}
	\item We propose a new problem to learn the \ac{SLF} based on \ac{A2A} path loss measurements, while, at the same time, \textit{steering} the original model into another one better represented by the data. This strategy results in a non-convex optimization problem, but it is marginally convex, i.e., the problem becomes convex if a subset of variables is fixed. In contrast to \cite{gutierrez2018nonparametric,romero2016blind,letter}, we  constrain the updates to remain close to the original model.
	\item The problem of learning the model derived from the \ac{TPT} would be extremely ill-posed, so we propose a non-linear kernel approach based on the \ac{RBF} similar to \cite{gutierrez2018nonparametric,romero2016blind,letter}.
	\item As in\cite{gutierrez2018nonparametric,letter,lee2017channel}, the structure representing the \ac{SLF} is assumed to be group-sparse, so we consider the least squares problem regularized by the elastic net \cite{zou2005regularization} for the reconstruction of the \ac{SLF}.
	\item Similar to the seeding paper of this work \cite{letter}, we define a majorizing function that \textit{upperbounds} the original objective function, i.e., the new function is grater or equal than the original one in its entire domain. This strategy is known from stochastic approximation \cite{slavakis2014stochastic} and has been exploited for online learning in different application domains \cite{mairal2010online,lee2017channel,razaviyayn2016stochastic}. Unlike in our previous work, we prove that the new function is indeed a surrogate of the original one, i.e., both functions tend to the same real value when the number of iterations grows towards infinity.
	\item We propose a novel online algorithm similar to the seeding publication in the sense that it is also a ``descent" version of the \ac{gAM}. In this case however, the method to update the model is the projected gradient descent, while the iterative procedure for the \ac{SLF} is based on the forward-backward splitting method \cite{combettes2011proximal}.
	\item As main contribution of this work, we study the convergence of the proposed algorithm both in the objective and in the arguments. To this end, we first prove some regularity of the functions involved. More precisely, we require local Lipschitz-continuity, uniformly strong convexity and uniform boundedness. After proving that the majorizing function is a surrogate of the original objective function under some reasonable assumptions, we show that the iterations of the online algorithm converge to a point in the set of stationary points of the original problem.
	\item Unlike in the seeding publication, we evaluate our algorithm firstly with synthetic data representing the Madrid scenario \cite{metis_D_6_1} for a \ac{V2V} network, and secondly with measurements generated with the \ac{GEMV2}  \cite{boban2014geometry}, a simulator which has been shown to generate \ac{V2V} path loss datasets very close to real measurements \cite{boban2014geometry}. In both cases, we show the gains of our hybrid model and data driven approach.
\end{enumerate}

\subsection{Paper structure}
This paper is structured as follows. In Sect. \ref{sec:preliminaries} we review the notation, definitions and facts that are extensively used to prove the main results in this study. In Sect. \ref{sec:system_model} we both introduce the framework of path loss learning based on \acp{TPT}, and state the problem. In Sect. \ref{sec:online_learning} we present our algorithmic solution to the online learning of \ac{A2A} path loss maps, and Sect. \ref{sec:convergence_analysis} deals with the analysis of the algorithm convergence. We conclude the study with the numerical evaluation of our proposed algorithm in Sect. \ref{sec:numerical_evaluation}.

\section{Mathematical Preliminaries}
\label{sec:preliminaries}

The objective of this section is to introduce the mathematical machinery required for this study. We further introduce notation and results in mathematics that are necessary to keep the presentation as self-contained as possible. 

Hereafter, we denote linear and non-linear operators with uppercase letters, vectors with bold lowercase letters, and sets and function classes with calligraphic letters. Given $n\in\Natural,$ $I_n$ represents the $n\times n$ identity matrix, $\E$ represents the expected value, $\otimes$ is the Kronecker product, $\odot$ is the Hadamard product, the superscript $(\cdot)^\top$ denotes the transpose, and $\Tr(A)$ is the trace of the matrix $A$.  We consider the Euclidean space $\mathcal{H}:=\real^n,$ which is a Hilbert space equipped with the inner product $(\forall\signal{x},\signal{y}\in\mathcal{H})~\langle\signal{x},\signal{y}\rangle:=\signal{x}^\top\signal{y}.$ The norms $\|\cdot\|_1$ and $\|\cdot\|_2$ are, respectively, the standard $\ell_1$ and $\ell_2$ norms in the Euclidean space, unless otherwise stated.
\begin{Def}[Marginal convexity]
	\label{def:marginal_convex}
	A continuous function of two variables $k:\setX\times\setY\rightarrow\real$ is considered marginally convex in its first variable if for every value of $\y\in\setY$, the function $k_{\y}:\setX\rightarrow\real$ is convex, i.e. for every $\signal{x}^1,\signal{x}^2\in \setX$, there exists a subgradient $\signal{g}\in\partial k_\y(\x)$ such that 
	$$ k_{\y}(\signal{x}^2)\geq k_{\y}(\signal{x}^1)+\langle\signal{g},\signal{x}^2-\signal{x}^1\rangle.$$
	In case $k$ was differentiable in $\signal{x}$, the subgradient can be substituted by the corresponding gradient: $$ k_{\y}(\signal{x}^2)\geq k_{\y}(\signal{x}^1)+\left \langle \nabla k_{\y}(\signal{x}^1) ,\signal{x}^2-\signal{x}^1\right \rangle. $$
	A similar argument can be made for the second variable $\y$.
\end{Def}
\begin{Def}[Marginal optimum coordinate]
	\label{def:marginally_opt}
	Let $k:\setX\times\setY\rightarrow\real.$ For any point $\y\in\setY,$ we say that $\tilde{\x}$ is a marginally optimal coordinate with respect to $\y,$ and use the shorthand $\tilde{\x}\in\textup{mopt}_k(\y),$ if $k(\tilde{\x},\y)\leq k(\x,\y)~\forall\x\in\setX.$ Similarly, for any $\x\in\setX,$ we say $\tilde{\y}\in\textup{mopt}_k(\x)$ if $\tilde{\y}$ is a marginally optimal coordinate with respect to $\x.$
\end{Def}
\begin{Def}[Bistable point]
	\label{def:bistable_point}
	Given a function $k:\setX\times\setY\rightarrow \real$, a point $(\x,\y)\in\setX\times\setY$ is considered a bistable point if $\y\in\textup{mopt}_k(\x)$ and $\x\in\textup{mopt}_k(\y),$ i.e. both coordinates are marginally optimal with respect to each other.
\end{Def}
\begin{Def}[Proximal operator]
	\label{def:prox_operator}
	The proximity operator of a function $k:\real^n\rightarrow \real$ is given by: 
	$$ \textup{prox}_{\gamma k }:\real^n\rightarrow\real^n,(\f)\mapsto \underset{\boldsymbol{y}\in \real^n}{\textup{argmin}}~k(\boldsymbol{y})+\gamma\|\f-\boldsymbol{y}\|_2^2,$$
	where $\gamma>0$ is the attraction parameter.
\end{Def}
\begin{Def}[Directional derivative]
	\label{def:direc_derivative}
	Let $h:\setX\rightarrow\real$ be a convex function, where $\setX\subseteq\real^p$ is a closed convex set. The directional derivative of the function $h$ at a point $\x\in\setX$ in the direction $\di\in\real^p$ is defined as 
	$$ h'(\x;\di)\triangleq\underset{t\downarrow 0}{\limit~\infi}~ \frac{h(\x+t\di)-h(\x)}{t},$$
	and we define $h'(\x;\di)\triangleq+\infty$ if $\x+t\di\notin\setX,~\forall t>0.$
\end{Def}
\begin{Def}[Stationary point]
	\label{def:stationary_point}
	Let $h:\setX\rightarrow \real$ be a function, where $\setX\subseteq\real^p$ is a convex set. The point $\x\in\real^p$ is a stationary point of $h$ if $$ h'(\x;\di)\geq 0,~\forall\di\in\real^p. $$
\end{Def}
\begin{Def}[Contraction mapping]
	\label{def:contraction_map}
	Let $\setE$ be an Euclidian space. Then $T:\setE\rightarrow\setE$ is said to be a contraction mapping if there exists $\kappa\in[0,1[$ such that 
	$$ (\forall \x,\y\in\setE)~\|T(\x)-T(\y)\|\leq\kappa\|\x-\y\|. $$
\end{Def}
\begin{Def}[Non-expansive mapping]
	\label{def:nonexpansive_map}
	Let $\set{D}$ be a nonempty subset of $\setE$ and let $T:\set{D}\rightarrow\setE.$ The operator $T$ is said to be non-expansive if
	$$(\forall \x,\y\in\set{D})~ \|T(\x)-T(\y)\|\leq\|\x-\y\|. $$
	Further, $T$ is said to be firmly non-expansive if $(\forall \x,\y\in\set{D})$
	$$ \|T(\x)-T(\y)\|^2+\|\bar{T}(\x)-\bar{T}(\y)\|^2\leq\|\x-\y\|^2, $$
	where $\bar{T}=\Id-T.$
\end{Def}
\begin{Def}[Equicontinuous function]
	\label{def:equi_func}
	Let $\setX$ and $\setY$ be two metric spaces, $d(\cdot,\cdot)$ their respective distance metric, and $\setF$ a family of functions mapping $\setX$ into $\setY.$ The family $\Theta$ is said to be equicontinuous at a point $\x_0\in\setX$ if for every $\varepsilon>0,$ there exists a $\gamma>0$ such that $d(f(\x_0),f(\x))<\varepsilon$ for all $f\in\Theta$ and all $\x$ such that $d(\x_0,\x)<\gamma.$ 
\end{Def}
\begin{Def}[Quasi-martingale]
	\label{def:quasi_martingale}
	A martingale is a stochastic process for which, at a particular time, the conditional expectation of the next value in the sequence, given all prior values, is equal to the present value, i.e., a stochastic process $X_1,X_2,...,X_2$ is said to be martingale if for a particular time instant $n$, we have $$\E[X_n]<\infty,$$and $$\E[X_{n+1}|X_1,...X_n]=X_n.$$ Further, a stochastic process is said to be quasi-martingale if it has a decomposition into the sum of a martingale process and a sequence of functions having almost every sample of bounded variation \cite{fisk1965quasi}.
\end{Def}
\begin{Def}[Filtration of a stochastic process]
	\label{def:filtraton}
	Consider a real valued stochastic process $\{X_n\}_{n=1}^\infty.$ For each $n$, we define the filtration of the stochastic process up to instant time $n$ as 
	$$ \set{N}_n:=\sigma(X_1,...,X_n),$$
	where $\sigma(X_1,...,X_n)$ denotes the $\sigma$-algebra generated by the random variables $X_1,...,X_n.$
\end{Def}

In the following, we state several theorems and lemmata necessary for the convergence analysis of our algorithms.
\begin{fact}[Bonnans and Shapiro theorem \cite{bonnans1998optimization}]
	\label{theo:bonnas}
	Let $f:\real^p\times\real^q\rightarrow\real$. Suppose that for all $\x\in \real^p$ the function $f(\x,\cdot)$ is differentiable, and that $f$ and $\nabla_u f(\x,\signal{u})$ the derivative of $f(\x,\cdot)$ are continuous on $\real^p\times\real^q$. Let $v(\signal{u})$ be the optimal value function $v(\signal{u})=\textup{min}_{\x\in\setC}f(\x,\signal{u})$, where $C$ is a compact subset of $\real^p.$ Then $v(\signal{u})$ is directionally differentiable. Furthermore, if for $\signal{u}_0\in\real^q,f(\cdot,\signal{u}_0)$ has a unique minimizer $\x_0$ then $v(\signal{u})$ is differentiable in $\signal{u}_0$ and $\nabla_u v(\signal{u}_0)=\nabla_u f(\x_0,\signal{u}_0).$
\end{fact}
\begin{fact}[A Corollary of Donsker's theorem \cite{van2000asymptotic}]
	\label{lemma:donsker}
	Let $F=\{f_\theta:\setX\rightarrow\real,\theta\in\Theta\}$ be a set of measurable functions indexed by a bounded subset $\Theta\in\real^d.$ Suppose that the functions $f_\theta\forall\theta\in\Theta$ are Lipschitz continuous, $\|f\|_\infty<K$, $\E_X[f(X)^2]\leq\delta^2$ for some $\delta>0$, and that the random elements $X_1,X_2,...$ are Borel-measurable. Define the empirical average $u_n(f)=\frac{1}{n}\sum_{i=1}^n f(X_i).$ Then, 
	$$\E\left[\underset{f\in F}{\textup{sup}}~\sqrt{n}\left|(u_n(f)-\E_X[f])\right|\right]\in\mathcal{O}(1).$$
\end{fact}
\begin{fact}[Lemma on positive converging sums \cite{razaviyayn2016stochastic}]
	\label{lemma:positive_congerging_sum}
	Let $(u_t)_{t=1}^\infty$ be a nonnegative sequence, i.e., $(\forall t\in\Natural)~u_t\geq0,$ and $\sum_{t=1}^\infty u_t/t<\infty.$ Furthermore, suppose that $(\forall t\in\Natural)~|u_{t+1}-u_t|\leq c/t$ for some $c>0.$ Then, $\limit_{t\rightarrow\infty}u_t=0.$
\end{fact}
\begin{fact}[Theorem on the sufficient condition of convergence for a stochastic process \cite{fristedt2013modern}]
	\label{theorem:quasi_martingales}
	Let $(\Omega,\setN,P)$ be a measurable probability space, $u_t$, for $t\geq 0,$ be the realization of a stochastic process and $\setN_t$ be the filtration determined by the past information at time $t$. Let 
	$$ \delta_t=\left\{\begin{matrix}
	1&~if ~\E[u_{t+1}-u|\setN_t]>0,\\ 
	0&~otherwise
	\end{matrix}\right. $$
	If for all $t,u_t\geq0$ and $\sum_{t=1}^\infty \E[u_{t+1}-u_t]<\infty,$ then $u_t$ is a quasi-martingale and converges almost surely. Moreover, $$ \sum_{t=1}^\infty |\E[u_{t+1}-u_t|\setN_t]|<+\infty~a.s. $$ 
\end{fact}
\begin{fact}[Glivenko-Cantelli theorem \cite{van2000asymptotic}]
	\label{theorem:Glivenko}
	Let $F$ be the cumulative distribution function of the stochastic process generating i.i.d. samples $X_1,X_2,...$, and let $F_n$ be its empirical cumulative distribution function after $n\in\Natural$ samples, given by $F_n(n)=\frac{1}{n}\sum_{i=1}^n 1[X_i\leq x].$ Then, $$ \|F_n-F\|_\infty\rightarrow 0. ~a.s.$$
\end{fact}
\begin{fact}[Arzelà–Ascoli theorem \cite{dunford1958linear}]
	\label{theorem:arzela}
	Consider a sequence of real-valued functions $(f_n)_{n\in\Natural}$ defined in a closed and bounded set $\setF.$ If this sequence is uniformly bounded and uniformly equicontinuous, then there exists a subsequence $(f_{n_k})_{k\in\Natural}$ that converges uniformly. The converse is also true, in the sense that if every subsequence of $(f_n)_{n\in\Natural}$ itself has a uniformly convergent subsequence, then $(f_n)$ is uniformly bounded and equicontinuous.
\end{fact}
\begin{fact}[Proposition on the existence of directional derivative \cite{mordukhovich2013easy}]
	\label{prop:exitence_direc_deriv}
	For any convex function $f:\real^n\rightarrow(-\infty,\infty]$ and any $\bar{\signal{x}}\in\mathrm{dom}(f),$ the directional derivative $f^'(\bar{\signal{x}};\di)$ exists for every direction $\di\in\real^n.$ Furthermore, if $\bar{\signal{x}}\in\mathrm{int}(\mathrm{dom}~f)$ then $f^'(\bar{\signal{x}},\di)$ is a real number for every $\di\in\real^n.$
\end{fact}
\begin{fact}[Mean value theorem of vector calculus \cite{rudin1964principles}]
	\label{theo:mean_value}
	Let $f:U\rightarrow\real$ be a differentiable function, where $U$ is a convex and open subset of $\real^k$. Let $\signal{a},\signal{b}$ be points in $U$, with $\signal{b}\succeq\signal{a}$. Then, there exists $\x\in]\signal{a},\signal{b}[$ such that 
	$$\langle\nabla f(\x),\signal{b}-\signal{a}\rangle=f(\signal{b})-f(\signal{a}).$$
\end{fact}
\begin{fact}[Firmly non-expansiveness of proximal operators \cite{combettes2005signal}]
	\label{lemma:non_expansiveness}
	Let $\varphi\in\Gamma_0(\setX)$, where $\Gamma_0$ is the class of lower semicontinuous convex functions from $\setX.$ Then $\textup{prox}_\varphi$ and $\textup{Id}-\textup{prox}_\varphi$ are firmly nonexpansive.
\end{fact}

\section{System Model and Problem Statement}
\label{sec:system_model}

Consider a two-dimensional area $\mathcal{A}\subset\real^2$. We model the long-term average path loss between any two points $\mathbf{x}_i,\mathbf{x}_j\in\mathcal{A}$ in logarithmic scale by:
\begin{equation}
\label{eq:pathloss}
\mathrm{pl}(\mathbf{x}_i,\mathbf{x}_j)=\mathrm{pl}_0+10\delta\log_{10}\left(\frac{||\mathbf{x}_i-\mathbf{x}_j||_2}{d_0}\right)+s(\mathbf{x}_i,\mathbf{x}_j)+\epsilon_s,
\end{equation}
where $\mathrm{pl}_0$ is the path loss at a reference distance $d_0$, $\delta>0$ is the path loss exponential decay, $s:\real^2\times\real^2\rightarrow\real$ represents the shadowing function between $\mathbf{x}_i$ and $\mathbf{x}_j$, and $\epsilon_s>0$ is a scalar that accounts for the error in the measurements.

As in \cite{hamilton2013propagation,romero2016blind,lee2017channel,gutierrez2018nonparametric}, we model the shadow fading with a \ac{TPT}. More precisely, we consider the shadowing to be modeled as follows:
\begin{equation}
\label{eq:discrete_shadow_fading}
s(\mathbf{x}_i,\mathbf{x}_j)= \sum_{p=1}^{P_xP_y}w(\phi_1(\mathbf{x}_i,\mathbf{x}_j),\phi_2(\mathbf{x}_{p},\mathbf{x}_i,\mathbf{x}_j))f(\mathbf{x}_{p}),
\end{equation}
where $w:\real\times\real\rightarrow\real$ is the window function; $f:\mathcal{A}\rightarrow\real$ is the \ac{SLF} function; the distance functions $\phi_1:\real^2\times\real^2\rightarrow\real_+,(\mathbf{x}_i,\mathbf{x}_j)\mapsto||\mathbf{x}_i-\mathbf{x}_j||_2$ and $\phi_2:\real^2\times\real^2\times\real^2\rightarrow\real_+,(\mathbf{x}_i,\mathbf{x}_j,\mathbf{x}_p)\mapsto ||\mathbf{x}_p-\mathbf{x}_i||_2+||\mathbf{x}_p-\mathbf{x}_j||_2$ are functions that measure the length of the direct link and the length of the path going through an intermediate point, respectively; $P_x$, $P_y$ is the number of horizontal and vertical pixels of the map, respectively, and $\mathbf{x}_{p}$ is the coordinate of the pixel $p\in \overline{1,P}$, with $P=P_xP_y$. Intuitively speaking, the shadowing between any two points in a map is potentially influenced by the SLF at any point of the map. This assumption is due to the multi-path nature of radio signals propagation. To capture these effects, the SLF is weighted by a window function that models the influence of each position on a link.


There are several models in the literature that exploit the concept of the Fresnel zone and try to model the window function in an statistical form \cite{agrawal2009correlated}. For example, the normalized elliptical model considers that an ellipsoid with foci at each node location determines the influence for each link in the area \cite{patwari2008effects}.
Mathematically speaking, the window function is defined according to the normalized elliptical model as follows \cite{hamilton2013propagation,wilson2010radio}:
\begin{equation}
\label{eq:elliptical_model}
w(\phi_1,\phi_2):=\left\{\begin{matrix}
0& \textup{if}~\phi_2>\phi_1+\eta/2\\ 
\frac{1}{\sqrt{\phi_1}}& \textup{otherwise},
\end{matrix}\right.
\end{equation}
with $\eta$ being in this case the signal wavelength. 

The normalized elliptical model assumes that all the points inside the ellipse have equal weight. Another model called inverse area elliptical model considers that some parts of the ellipse have a greater contribution than others \cite{hamilton2013propagation}. The reasoning behind this is that signal paths closer to the edge of the ellipse travel longer distances than those closer to the line of sight, so their contribution to the shadowing should be lower. Mathematically, the model is described as follows:
\begin{equation*}
w(\phi_1,\phi_2):=\left\{\begin{matrix}
0& \!\!\!\!\!\!\!\!\!\!\!\textup{if}~\phi_2>\phi_1+\eta/2\\ 
\textup{min}(\Gamma(\phi_1,\phi_2),\Gamma(\phi_1,\phi_1+\nu))& \textup{otherwise},
\end{matrix}\right.
\end{equation*}
with $\nu>0$ being an user-selected parameter and
$$ \Gamma(\phi_1,\phi_2)=\frac{4}{\pi\phi_2\sqrt{\phi_2^2-\phi_1^2}}. $$

We now proceed to write \eqref{eq:discrete_shadow_fading} in a matrix form. To this end, let $\signal{f}=[f_1,...,f_{P}]^{\top}\in\real^{P}$, where $f_{p}:=f(\mathbf{x}_p)$. Assuming channel reciprocity in the path loss between any two points, the total number of links $T$ in a map of $P$ pixels is $T:=P(P-1)/2,$ and the index set $\mathcal{M}$ of all links is $\mathcal{M}:=\overline{1,...,T}$. We define a bijective mapping $$B:\mathcal{P}\times\mathcal{P}\rightarrow \mathcal{M}:~m\mapsto P(j-1)+i$$ 
that maps any two indexes $i,j \in \mathcal{P}$ onto a link index $m\in\mathcal{M}$.
By $W\in\real^{T\times P}$ we define a matrix containing all possible weight values of the map, where  $w_{m,p}:=w(\phi_1(\mathbf{x}_i,\mathbf{x}_{j}),\phi_2(\mathbf{x}_{i},\mathbf{x}_j,\mathbf{x}_{p}))$ with $m:=B(i,j)$. The shadow fading vector $\signal{s}\in\real^{T}$ is generated by stacking all possible values of $s(\mathbf{x}_i,\mathbf{x}_j)$, such that:
\begin{equation}
\label{eq:vector_s}
\signal{s}=W\signal{f}.
\end{equation}


Consider that measured path loss values arrive at a central entity at different time instants. Let $\shahat{_t}\in\real^{M_t}$ be the (noisy) shadowing measurements acquired at time instant $t\in\Natural$, and $\Omega_t:=\{\omega_t^1,...,\omega_t^{M_t}\}$ is the measurements index set with cardinality $M_t$ \cite{kim2010cooperative,dall2011channel}. For the sake of simplicity, we assume that $M_1=M_2=\cdots=M$ for the reminder of this manuscript. The elements in $\shahat_{t}$ represent a selection of all elements contained in $\shahat\in\real^{T}$, which is in turn a vector containing all possible (noisy) measurements in a map. Analogously, the matrix $W_{t}\in\real^{M\times P}$ denotes the weight matrix whose rows correspond to the link measurements received at time $t$. Let $(\forall i\in\Omega_t)~\w_i\in\real^P$ be the rows of $W_t$. The convex set $\setC_i$ is defined as the $\ell_2$-ball of radius $r$ around the vector $\w_i$:
$$ \setC_i:=\{\x\in\real^P|\|\x-\w_i\|_2\leq r\}. $$
Through the sets $\setC_i$ we add the context information related to a certain model such as the elliptical model, or the normalized elliptical previously introduced. But, instead of assuming that the window function has to follow one of these or any other model, we encode this information in the form of a constraint convex set to our optimization problem. Such an approach allows us to include information coming from a model, since a solution to our problem is expected to be close to the model, but it also gives us \textit{some} freedom in case the physical environment is not precisely expressed by it. Note that we have restricted our constraint sets to the $\ell_2$-ball of radius $r$ and centered in  $\w_i$. One can change or add the constraint sets to include more context information, provided that the sets are convex ones and their intersection is not empty \cite{artacho2019douglas}. 

The proposed approach is based on the assumption that the \ac{SLF} vector $\signal{f}$ is a block-sparse vector \cite{eldar2010}. The block-sparsity of $\signal{f}$ is justified by the fact that most pixels of a map represent the free space, whose absorption value is negligible compared to the absorption of solid bodies, and therefore assumed to be zero. Further, non-zero entries of $\signal{f}$ are those belonging to walls and other physical structures, therefore they are concentrated in groups.

In  light of the above assumption, an intuitive approach to estimate the \ac{SLF} is to minimize the least squares error regularized by the elastic net to improve block-sparsity.

Previous studies such as \cite{romero2016blind,gutierrez2018nonparametric} in this particular application domain have shown that attempts at minimizing with respect to $W$ fail to give good results because the problems are in general severely ill-posed. To address this limitation, we impose additional structure on $W$ by considering a non-linear kernel approach similar to \cite{gutierrez2018nonparametric}.

With some abuse of notation, we define the vector $\signal{c}:=[c_{1},...,c_{M}]^{\top}\in\real^{M}$, where $c_{m}:=\phi(\mathbf{x}_i,\mathbf{x}_j)$. Similarly, denote $\signal{d}:=[d_{1,1},...,d_{M,1},d_{1,2},...,d_{M,P}]^{\top}\in\real^{M P}$, where $d_{m,p}:=\phi_2(\mathbf{x}_i,\mathbf{x}_j,\mathbf{x}_p)$. Define $\signal{\Phi}_{m,p}:=[c_m,d_{m,p}]^{\top}\in\real^2_+$ as a two-dimensional vector with arbitrary $c_m,d_{m,p}$ stacked together;  and we assume that the window function can be written as a function of a positive definite kernel in the following form:
\begin{equation}
\label{eq:w_of_phi}
	w(\signal{\Phi})=\sum_{m=\omega_t^1}^{\omega_t^M}\sum_{p=1}^{P}\alpha_{q(m,p)} \kappa(\signal{\Phi}_{m,p},\signal{\Phi}),
\end{equation}
where $\alpha_{q(m,p)}\in\real$ are appropriate scalars to be determined, $\omega_t^1,...,\omega_t^M$ are the ordered elements of $\Omega_t$, and $q(m,p)$ is an index obtained as $q(m,p):=P(m-1)+p$. In particular, in this study we use \ac{RBF} as kernel:
$$
\kappa_{q,q'}:=\kappa(\signal{\Phi}_{q},\signal{\Phi}_{q'})=\exp\left(-\frac{\left \|\signal{\Phi}_{q}-\signal{\Phi}_{q'} \right \|^2_2}{2\sigma^2}\right),
$$
where $\sigma>0$ is the width of the kernel.
Define the kernel matrix $K_t\in \real^{M P\times M P}$ with the $(q,q')$ element given by  $\kappa_{q,q'}$, and also define the vector $\alphas_t$ as $\alphas_t=[\alpha_{q_1},\cdots,\alpha_{q_{M P}}]^{\top}\in\real^{MP}$.

Let $j\in\Omega_t$, and define the convex set 
$$ \setQ_{t,j}:=\{\x\in\real^P| \|\x-\av_{t,j}\|_2\leq r\}, $$
where $\av_{t,j}\in\real^P$ is a vector given by $\av_{t,j}:=[\alphas_t]_j$, with $[\cdot]_j$ representing the selection of entries $P(j-1)+1$ until $Pj$ of a vector.
The set in the $\alphas$-parameter space $\setC_t^{\alpha}$ is defined at time $t$ as
$$ \setC_t^{\alpha}:=\bigtimes_{j=1}^M\setQ_{t,j} \subset\real^{MP},$$
which is also convex, since the Cartesian product of convex sets is a convex set \cite{artacho2019douglas}. Let $\x\in\real^P$. The projection of $\x$ onto the set  $\setQ_{t,j}$ is given by:
$$ \Pi_{\setQ_{t,j}}(\x)=\left\{\mathbf{q}\in\setQ_{t,j}|\mathbf{q}=r_{\textup{min}}\frac{\x-\av_{t,j}}{\|\x-\av_{t,j}\|_2}+\av_{t,j}\right\}, $$
with $r_{\textup{min}}=\min(r,\|\x-\av_{t,j}\|_2)$.
We can construct the projection of a vector $\alphas\in\real^{MP}$ onto $\setC_t^\alpha$ as \cite[Proposition 8]{artacho2019douglas}:
\begin{equation}
	\label{eq:proj_C_j_a}
	\Pi_{C_t^\alpha}(\alphas)=\bigtimes_{j=1}^M \Pi_{\setQ_{t,j}}(\x_j),
\end{equation}
where $\x_j:=[\alphas]_j\in\real^P$.

Considering the non-linear kernel approach, we can write an optimization problem over $\f,\alphas_{1},...,\alphas_t$ as follows:
\begin{multline}
\label{eq:problem_3}
\underset{\alphas_{\tau}\in \setC^{\alpha}_{\tau},\f\in\real^P}{\textup{minimize}}\frac{1}{t}\sum_{\tau=1}^t\left \|\shahat_{\tau}-A_{\signal{f}}K_{\tau}\signal{\alpha}_{\tau} \right \|^2_2
\\+ \lambda_1\left\|\signal{f}\right\|_1+\lambda_2\left \|\signal{f}  \right \|^2_2+\lambda_3\|\alphas_\tau\|^2_2,
\end{multline}
where $\lambda_1$, $\lambda_2$ and $\lambda_3$ are nonnegative regularization parameters \cite{zou2005regularization}, $A_{\signal{f}}=I_{M}\otimes \f^{\top}\in \real^{M\times MP}$ is the Kronecker product between the identity matrix $I_{M}\in\real^{M\times M}$ and $\f^{\top}$, and the term $\lambda_3\|\alphas_\tau\|^2_2$ is included to guarantee stability of the solutions. 

Problem \eqref{eq:problem_3} can be viewed as the minimization of the empirical cost of a model-constrained least squares regression regularized with the elastic net. Note however that \eqref{eq:problem_3} is not jointly convex in $\signal{f},\signal{\alpha}_{1},...,\signal{\alpha}_{t}$, but it is convex in $\signal{f}$ if $\signal{\alpha}_{1},...,\signal{\alpha}_{t}$ are fixed, and vice versa \cite{jain2017non}. As a result, we consider an alternating minimization strategy to address \eqref{eq:problem_3} where, at time $t$ of arrival of new measurements, one set of variables is updated while the remaining ones are kept constant. This process is carried out until a stopping criterion is met.

\section{Online Path Loss Learning}
\label{sec:online_learning}

In scenarios like the one presented in the previous section, one is in general interested in obtaining an estimate $\check{\f}_t$ of $\f$ at time $t$, instead of waiting for all measurements to arrive before solving Problem \eqref{eq:problem_3}. In these cases, a standard alternative is to try to optimize an empirical cost function \cite{mairal2010online}. To this end, we first define $(\forall t\in\Natural)~l(\shahat_t,\f)$ as the optimal value of the partial optimization problem with respect to $\alphas_t$:
\begin{equation}
\label{eq:cost_alpha}
	l_t(\f)\triangleq l(\shahat_t,\f)=\underset{\alphas_t\in\setC^{\alpha}_t}{\textup{min}} \frac{1}{2}\left \|\shahat_t-A_{\f}K_{t}\alphas_t \right \|^2_2+\lambda_3\|\alphas_t\|^2_2,
\end{equation}
and then we define the empirical cost function of the \ac{SLF} problem as
\begin{equation}
	\label{eq:empirical_cost}
	h_t(\f)=\frac{1}{t}\sum_{\tau=1}^{t}\ell_\tau(\f)\triangleq l_\tau(\f)+\lambda_1\|\f\|_1+\frac{1}{2}\lambda_2\|\f\|_2^2.
\end{equation}
The problem of online \ac{SLF} learning corresponds to the minimization of \eqref{eq:empirical_cost}:
\begin{equation}
	\label{eq:min_empirical_cost}
	\underset{\f\in\real^P}{\textup{minimize}}~h_t(\f).
\end{equation}
Note that \eqref{eq:problem_3} and \eqref{eq:min_empirical_cost} are equivalent problems in the sense that their set of minimizers are identical. However, given a finite training set, one should not spend too much time on accurately minimizing the empirical cost, since it is only an approximation of the expected cost and might not provide good solutions, especially when $t$ is small \cite{bottou2008tradeoffs}. Another limitation of directly minimizing \eqref{eq:empirical_cost} is the fact that the complexity increases with the acquired number of samples, making the approach unsuitable for online settings. Therefore, our interest lies on the minimization of the expected cost of $h_t(\f)$:
\begin{equation}
	\label{eq:h}
	h(\f)\triangleq\E_\shahat[\ell_t(\f)]=\underset{t\rightarrow \infty}{\textup{lim}}h_t(\f)~a.s.,
\end{equation}
where the expectation, which for now is supposed to be finite, is taken relative to the probability distribution $p(\shahat)$ of the measurements. Later in Section \ref{sec:convergence_analysis} we will state the necessary conditions for \eqref{eq:h} to be true.

\subsection{Addressing the elastic net subproblem}

Instead of minimizing the empirical cost in \eqref{eq:empirical_cost}, we propose the minimization at time instant $t$ of a new function $(\forall\alphas_1,...,\alphas_t\in\real^{MP})$
\begin{equation}
	\label{eq:surrogate}
	\check{h}_t(\f)=\frac{1}{t}\sum_{\tau=1}^t\frac{1}{2}\left \|\shahat_{t}-A_{\signal{\alpha}_{\tau}}\signal{f} \right \|^2_2+\lambda_1\left\|\signal{f}\right\|_1+\frac{1}{2}\lambda_2\left \|\signal{f}  \right \|^2_2,
\end{equation}
where $A_{\signal{\alpha}_{\tau}}=\sum_{n=1}^{M}\signal{e}_n\otimes(\signal{\alpha}_{\tau}^{\top}K_{\tau})\in\real^{M \times P}$, and $\signal{e}_n\in\real^{M}$ is a unitary vector with all zeros but the $nth$ entry one.

The motivation behind this approach lays on the fact that $\check{h}_t$ is convex in $\f$, and also because one can readily show that it upperbounds the empirical cost function $h_t(\f)$. Indeed, we will prove in Section \ref{sec:convergence_analysis} that $\check{h}_t$ acts as a surrogate of $h_t$, i.e. $\check{h}_t(\f)$ and $h_t(\f)$ converge to the same limit when $t\rightarrow\infty$.


We can rewrite $\check{h}_t(\f)$ in a more convenient way for our online algorithm in the following way:
\begin{equation}
\label{eq:h_check_t_modified}
\check{h}_t(\f)\!=\!\frac{\Tr(\bar{A}_t \f\f^\top\!)\!\!-\!\!\Tr(\textbf{b}_t\f^\top\!)}{2t}\!+\!\frac{1}{2}\lambda_2\|\f\|^2_2+\lambda_1\|\f\|_1+C,
\end{equation}
where $C=\Tr(\shahat_t\shahat_t^\top)/2t$ is a constant, $\bar{A}_t=\sum_{\tau=1}^t A_{\alphas_\tau}^\top A_{\alphas_\tau}\in\real^{P\times P}$, and $\textbf{b}_t=\sum_{\tau=1}^t A^\top_{\alphas_\tau}\shahat_\tau\in\real^P.$ Note that the structures $\bar{A}_t$ and $\textbf{b}_t$ do not change size with increasing $t$. This suggests an algorithm in which, at time $t$, we keep track and update  $\bar{A}_t$ and $\textbf{b}_t$, and a new estimate $\check{\f}_t$ is found after minimizing function $\check{h}_t(\f)$ in \eqref{eq:h_check_t_modified} w.r.t. $\f$. Note that $(\forall\alphas_1,...,\alphas_t\in\real^{MP})$ the function $\check{h}_{t}$ is coercive, proper and strongly convex, therefore 
\begin{equation}
\label{eq:f_star_update}
\check{\f}_{t}\in\textup{argmin}_{\f}~\check{h}_{t}(\f)
\end{equation}
exists and is unique \cite{bauschke2011convex}.

The function $\check{h}_t$ can be expressed as the sum of two functions $g_1$ and $g_2$, where $g_1$ is convex and differentiable, while $g_2$ is convex but non-smooth. More precisely, define 
\begin{equation}
	\label{eq:g_1}
	g_1\!:\!\real^{P}\!\rightarrow\!\real\!:\!\f\!\mapsto\! \frac{1}{2t}\!\left(\!\Tr(\bar{A}_t \f\f^\top)\!-\!\Tr(\textbf{b}_t\f^\top)\!+\!t\lambda_2\|\f\|^2_2\!\right)\!+\!C,
\end{equation}
and 
\begin{equation}
	\label{eq:g_2}
	g_2:\real^{P}\rightarrow\real:\f\mapsto \lambda_1||\f||_1.
\end{equation} 
The problem $\textup{min}_{\f}~\check{h}_{t}(\f)$ can then be formulated as:
\begin{equation}
\label{eq:canonical_f_problem}
\underset{\f}{\textup{min.}} ~g_1(\f)+g_2(\f).
\end{equation}
This kind of problems are well understood and there is a plethora of algorithms to solve them \cite{wang2014low,gandy2011tensor,bauschke2011convex,combettes2011proximal}. We propose using the forward-backward splitting method \cite{combettes2011proximal} due to its good performance compared to other methods in this particular application domain. It can be shown \cite{combettes2011proximal} that, if $g_1$ is Lipschitz-differentiable, Problem \eqref{eq:canonical_f_problem} admits one solution and that, for certain $\gamma\in]0,\epsilon[$ and $\epsilon>0$, its solution is characterized by the fixed point equation 

\begin{equation}
\label{eq:fixed_point_eq_f_problem}
\f=\textup{prox}_{\gamma g_2}(\f-\gamma\nabla g_1(\f)),
\end{equation}
where $\textup{prox}_{\gamma g_2}$ is the proximal operator of $g_2$ with attracting factor $\gamma$. 


An iterative solution to \eqref{eq:canonical_f_problem} is then given by

\begin{equation}
\label{eq:its_f_problem}
\f^{(n+1)}=\textup{soft}_{\lambda_1}\left(\f^{(n)}+\gamma\left(\frac{1}{2t}\textbf{b}^\top_t -\frac{1}{t}\bar{A}_t\f^{(n)}-\lambda_2\f^{(n)}\right)\right),
\end{equation}
where $\textup{soft}_{\lambda_1}(\cdot)$ is the soft thresholding function with threshold $\lambda_1$, and $n\in\Natural$ is the iteration index.

\subsection{Addressing the constraint least squares subproblem}
Unlike the minimization of $\f$ in Problem \eqref{eq:problem_3}, the minimization over $\alphas_1,...,\alphas_t$ is not coupled in the summation of functions through its variables, meaning that we can separate Problem \eqref{eq:problem_3} with respect to $\alphas_1,...,\alphas_t$, while keeping $\f$ fixed. 

The problem of minimizing $\alphas_t$ is defined $(\forall\f\in\real^P)$ as follows:

\begin{equation}
\label{eq:least_squares_problem}
\underset{\alphas_t\in\setC^{\alpha}_t}{\textup{min}}~k_t(\alphas_t)\triangleq\frac{1}{2}\left \|\shahat_t-A_{K_t}\alphas_t \right \|^2_2+\lambda_3\|\alphas_t\|^2_2,
\end{equation}
where $A_{K_t}=A_\f K_t\in \real^{M\times MP}.$ Problem \eqref{eq:least_squares_problem} has $(\forall t\in\Natural)$ a unique solution because $(\forall\alphas_t\in\real^{MP})~k_t(\alphas_t)$ is a quadratic function and the lowest eigenvalue of its Hessian is at least $\lambda_3$. Such solution can be attained with the projected gradient method:

\begin{equation}
	\label{eq:projected_gradient_its}
	\alphas_t^{(n+1)}\!=\!\Pi_{\setC^{\alpha}_t}\!\left(\alphas_t^{(n)}\!+\!\mu\! \left(A_{K_t}^\top\shahat\!-\!A_{K_t}^\top A_{K_t}\alphas_t^{(n)}\!-\!\lambda_3\alphas_t^{(n)}\right)\right)\!,
\end{equation}
where $\Pi_{\setC^{\alpha}_t}$ is the projection onto the convex set $\setC^{\alpha}_t $ given by \eqref{eq:proj_C_j_a}, $\mu>0$ is the step size, and $n\in\Natural$ is the iteration index. 
The reason for the selection of an iterative method such as the projected gradient to solve \eqref{eq:least_squares_problem} instead of solving the dual problem will become apparent in the next section, but for now we mention the need of the intermediate values $\alphas_t^{(n)}$ for our algorithmic solution.



\subsection{Algorithmic Solution}
\label{sec:algorithmic_solution}
The missing piece in the online \ac{SLF} learning problem is the combination in an algorithm of the estimates $\check{\f}_t$ in \eqref{eq:f_star_update} with the iterative solutions of $\f^{(n)}$ and $\alphas^{(n)}_t$ in \eqref{eq:its_f_problem} and \eqref{eq:projected_gradient_its}, respectively. 

One important caveat of the online algorithm is that we implement a ``descent" version of the outlined alternating minimization process. This means that instead of running the iterations in \eqref{eq:its_f_problem} until a stopping criterion is met, and then proceed with the iterations in \eqref{eq:projected_gradient_its} again until improvements are small enough, we take only one step at a time of the iterations in \eqref{eq:its_f_problem}, and another step of iterations in \eqref{eq:projected_gradient_its}, alternatively until a combined stopping criterion is met. The rationale behind this is that the improvement from $\f_t^{(1)}$ to $\f_t^{(N)}$ and from $\alphas_t^{(1)}$ to $\alphas_t^{(N)}$ might not be relevant enough to justify finding the optimal solutions of the two convex sub-problems in $\alphas_t$ and $\f_t$ alternatively. Indeed, simulations for this particular application have consistently shown better performance and shorter execution time with the ``descent" strategy.

Consider that the samples $\shahat_1,\shahat_2,...$ are i.i.d samples drawn from a common distribution $p(\shahat_t)$. In order to guarantee stability of the iterations, we need to choose the update parameters $\mu$ and $\gamma$ for every iteration index $n$ before updating $\alphas_t$ and $\f_t$, respectively. More precisely, we need to choose $\mu$ ($\gamma$) strictly smaller than the Lipschitz constant of $\nabla k_t(\alphas)$ ($\nabla g_1(\f)$), given by $L_k=\|A_\f^\top A_\f+\lambda_3 I_{MP}\|_2$ ($L_g=\|\bar{A}_t+\lambda_2 I_{P}\|_2$), where $\|\cdot\|_2$ here is the spectral norm of a matrix. Section \ref{sec:convergence_analysis} deals with the convergence analysis of the algorithm and provides formal proof for the selection of $\mu$ and $\gamma$. One important element to guarantee the convergence of the algorithm is the selection of an $\epsilon$ such that $\epsilon\in]0,1[$.

Finally, since $\check{h}_t$ is expected to be close to $\check{h}_{t-1}$ for large values of $t$, so are under suitable conditions $\check{\f}_t$ and $\check{\f}_{t-1}$, which makes it efficient to use $\check{\f}_{t-1}$ as ``warm" initialization for computing $\check{\f}_t$. Our procedure is summarized in Algorithm \ref{alg:online}.

\begin{algorithm}
	\SetKwInOut{Input}{Input}
	\SetKwInOut{Fix}{Fix}
	\SetKwInOut{Init}{Init}
	\SetKwInOut{Output}{Output}
	\Input{$\f^{(0)}_1$}
	\Fix{$\lambda_1,\lambda_2,\lambda_3,\epsilon$}
	\Init{$\bar{A}_0\in\real^{P\times P}\gets 0$, $\textup{b}_0\in\real^{P}\gets 0$}
	\For{$t=1,...,t_{\textup{max}}$}{
		$n\gets 0$\\
		Draw $\shahat_t$ from $p(\shahat)$\\
		Random Init $\alphas^{(n)}_t$\\
		\While{\textup{stopping criterion NOT met}}{
			Choose $\mu\in]0,(1-\epsilon)/L_k]$\\
			Update $\alphas^{(n+1)}_t$ according to \eqref{eq:projected_gradient_its}\\
			$\bar{A}_t\gets \bar{A}_{t-1}+A_{\alphas^{(n+1)}_t}^\top A_{\alphas^{(n+1)}_t}$\\
			$\textbf{b}_t\gets \textbf{b}_{t-1}+A_{\alphas^{(n+1)}_t}^\top\shahat_t$\\
			Choose $\gamma\in]0,(1-\epsilon)/L_g]$\\
			Update $\f^{(n+1)}_t$ according to \eqref{eq:its_f_problem}
			
		}
		Update $\check{\f}_t\gets \f^{(n)}_t$ \\
		Warm Init $\f^{(0)}_{t+1}\gets \check{\f}_{t}$\\
	}
	\Output{$\alphas_1,...,\alphas_{t_{\textup{max}}},\check{\f}_{t_{\textup{max}}}$}
	\caption{Online algorithm}
	\label{alg:online}
\end{algorithm}

\subsection{Complexity}

The complexity of Alg. \ref{alg:online} is dominated by the two matrix multiplications (one to compute $\alphas^{(n+1)}_t$, the other one to compute $\f^{(n+1)}_t$) required in each iteration $n$, which we assume $\mathcal{O}(n^3)$. We also assume that the stopping criterion in both cases is given by a maximum number of iterations $N$. The complexity is then given by $\mathcal{O}(t_{\textup{max}}N(P^3M^3+P^2M))$. Because the complexity is dominated by the first term ($ P^3M^3$), we see that it scales linearly with the total number of iterations ($t_{\textup{max}}N$), and cubicly for the number of pixels in the map ($P$) times the number of samples acquired at $t$ ($M$). 

\section{Convergence Analysis}
\label{sec:convergence_analysis}
In this section, we provide a convergence analysis of our proposed algorithm both in the objective and in the arguments. We focus our analysis on a modified version of Alg.~\ref{alg:online}, presented in Alg.~\ref{alg:alt_minimization}. Algorithm \ref{alg:alt_minimization} represents the non-descent version of Alg.~\ref{alg:online}, in the sense that, at time $t$, both sub-problems in $\alphas_t$ and $\f$ are solved sequentially, instead of alternating step by step between the two sub-problems as in Alg.~\ref{alg:online}.  We begin with the necessary (and reasonable) assumptions for the convergence of both algorithms. Several lemmata are then presented as intermediate statements to achieve the first of the two most relevant results, which is the proof of convergence of Alg.~\ref{alg:alt_minimization} to a stationary point of Problem \eqref{eq:min_empirical_cost}. After this, we trace back the connection with the original Alg.~\ref{alg:online} and also prove its convergence in the arguments.

\begin{algorithm}
	\SetKwInOut{Input}{Input}
	\SetKwInOut{Fix}{Fix}
	\SetKwInOut{Init}{Init}
	\SetKwInOut{Output}{Output}
	\Input{$\f_0$}
	\Fix{$\lambda_1,\lambda_2,\lambda_3$}
	\Init{$\bar{A}_0\in\real^{P\times P}\gets 0$, $\textup{b}_0\in\real^{P}\gets 0$}
	
	\For{$t=1,...,t_{\textup{max}}$}{
		Draw $\shahat_t$ from $p(\shahat)$\\
		$\alphas_t\gets \underset{\alphas_t\in\setC^{\alpha}_t}{\textup{argmin}} \frac{1}{2} \|\shahat_t-A_{\f_{t-1}}K_{t}\alphas_t  \|^2_2$\\
		$\bar{A}_t\gets \bar{A}_{t-1}+A_{\alphas_t}^\top A_{\alphas_t}$\\
		$\textbf{b}_t\gets \textbf{b}_{t-1}+A_{\alphas_t}^\top\shahat_t$\\
		$\f_t\gets \underset{\f}{\textup{argmin}}~\frac{1}{t}\sum_{\tau=1}^t\frac{1}{2} \|\hat{\signal{s}}_{\tau}-A_{\signal{\alpha}_{\tau}}\f  \|^2_2+\lambda_1\|\f\|_1+\frac{1}{2}\lambda_2 \|\f \|^2_2$ \\
		
		Update $\check{\f}_t\gets \f_t$ \\
		Warm Init $\f_{t+1}\gets \check{\f}_{t}$\\
	}
	\Output{$\alphas_1,...,\alphas_{t_{\textup{max}}},\check{\f}_{t_{\textup{max}}}$}
	\caption{Alternating minimization algorithm}
	\label{alg:alt_minimization}
\end{algorithm}


\begin{assumption}
	\label{ass:iid}
	The samples $\shahat_1,\shahat_2,...$ are i.i.d. samples drawn from a common distribution $p(\shahat)$ with compact support $\chi$.
\end{assumption}

\begin{assumption}
	\label{ass:compact_set_f_it}
	Let $\setF\subset\real^p$ be a convex, compact and non-empty set. We assume that the iterates $(\f_t)_{t\in\Natural}$ are in $\setF$. 
\end{assumption}


We now proceed to state our main results, namely, that $(\forall \f\in\setF)~h(\f)$ defined in \eqref{eq:h} exists almost surely, that $\check{h}_t$ acts asymptotically as a surrogate function of $h$, and that both Alg.~\ref{alg:online} and \ref{alg:alt_minimization} converge to a stationary point of $h$ asymptotically.
To facilitate the analysis, let us define the following functions:
\begin{equation}
	\label{eq:l_1_t}
	l^1_t(\f):=l_t(\f)+\frac{1}{2}\lambda_2\left \|\signal{f}  \right \|^2_2
\end{equation}
\begin{equation}
\label{eq:l_hat_1_t}
\hat{l}^1_t(\f):=\frac{1}{2}\left \|\hat{\signal{s}}_{t}-A_{\signal{\alpha}_{\tau}}\signal{f} \right \|^2_2+\frac{1}{2}\lambda_2\left \|\signal{f}  \right \|^2_2
\end{equation}
\begin{equation}
\label{eq:l_hat_t}
\hat{l}_t(\f):=\hat{l}^1_t(\f)+g_2(\f)
\end{equation}
\begin{equation}
\label{eq:h_1_t}
h^1_t(\f):=\frac{1}{t}\sum_{\tau=1}^{t} l^1_\tau(\f)
\end{equation}
\begin{equation}
	\label{eq:h_check_1_t}
	\check{h}^1_t(\f):=\frac{1}{t}\sum_{\tau=1}^{t} \hat{l}^1_\tau(\f)
\end{equation}

The following lemma states necessary properties of the previously defined functions. 
\begin{lemma}
	\label{lemma:regularity}
	Under Assumptions \ref{ass:iid}, \ref{ass:compact_set_f_it}, the following is true:
	\begin{enumerate}[label=(\roman*)]
		\item \textup{(Continuous differentiability)} $(\forall t\in\Natural)(\forall\shahat_t\in\chi)$ the functions $l_t$ defined in \eqref{eq:cost_alpha}, and $l^1_t$ defined in \eqref{eq:l_1_t} are continuously differentiable.
		\item \textup{(Local Lipschitz-continuity)} $(\forall t\in\Natural)(\forall \f_1,\f_2\in\setF)(\forall\shahat_t\in\chi)$ let $\hat{l}_t$ be the function defined in \eqref{eq:l_hat_t}. Then,
		$$ |\hat{l}_t(\f_1)-\hat{l}_t(\f_2)|\leq c_1\|\f_1-\f_2\|, $$
		where $c_1>0$ is a constant.
		\item \textup{(Uniformly strong convexity)} $(\forall t\in\Natural)(\forall \f_1,\f_2\in\setF)(\forall\shahat_t\in\chi)(\forall \signal{g}\in\partial \check{h}(\f_2))$ let $\check{h}_t$ be the function defined in \eqref{eq:surrogate}. Then,
		$$\check{h}_t(\f_1)-\check{h}_t(\f_2)\geq \signal{g}^\top(\f_1-\f_2)+\frac{c_2}{2}\|\f_1-\f_2\|^2,$$ 
		where $c_2>0$ is a constant.
		\item \textup{(Uniform boundedness)} $(\forall t\in\Natural)(\forall\shahat_t\in\chi)(\forall \f\in\setF)$ let $l^1_t$, $\hat{l}^1_t$, $g_2$ and $\hat{l}_t$ be the functions defined in \eqref{eq:l_1_t}, \eqref{eq:l_hat_1_t}, \eqref{eq:g_2} and \eqref{eq:l_hat_t}, respectively. There exist constants $K_3>0,$ $K_2>0$ and $K_1>0$ s.t. 
		$$ |l^1_t(\f)|\leq K_3,~\|\nabla l^1_t(\f)\|\leq K_3,~\|\nabla^2 l^1_t(\f)\|\leq K_3, $$
		$$ |\hat{l}^1_t(\f)|\leq K_3,~\|\nabla \hat{l}^1_t(\f)\|\leq K_3,~\|\nabla^2 \hat{l}^1_t(\f)\|\leq K_3, $$
		$$|g_2(\f)|\leq K_2,$$ 
		$$ (\forall \di\in\real^P~s.t.~\f+\di\in\setF)~|g_2'(\f;\di)|\leq K_2\|\di\|, $$ 
		$$|\hat{l}_t(\f)|\leq K_1.$$
	\end{enumerate}
\end{lemma}
\begin{proof}
	(i) Consider the problem $$\underset{\alphas\in\setC^{\alpha}}{\textup{min}}\frac{1}{2}\|\shahat_t-A_fK_t\alphas\|^2_2+\frac{1}{2}\lambda_3\|\alphas\|^2_2.$$ Under Assumptions 1 and 2, the solution $\alphas^*$ to the problem exists and is unique, since $\setC^{\alpha}$ is a convex set, and $\shahat_t\in\chi$ and $\f\in\setF$ belong to compact sets. We can apply Fact \ref{theo:bonnas}, which directly gives us the continuous-differentiability of $l_t$. The fact that $l^1_t$ is also continuously differentiable follows since $l^1_t$ is the sum of two continuous and differentiable functions.
	
	(ii) Let $\f_1,\f_2\in\setF$, and notice that $\|\f_1\|^2,\|\f_2\|^2\leq K$ for some constant $K>0$, since $\f_1,\f_2$ are in a compact set. Using the associativity property of matrix multiplication, we can easily verify that
	\begin{align*}
		&\f_1^\top A_{\alphas_t}^\top A_{\alphas_t}\f_1-\f_2^\top A_{\alphas_t}^\top A_{\alphas_t}\f_2\\&=(\f_1-\f_2)^\top A_{\alphas_t}^\top A_{\alphas_t}\f_2\\&+\f_2^\top A_{\alphas_t}^\top A_{\alphas_t}(\f_1-\f_2)+(\f_1-\f_2)^\top A_{\alphas_t}^\top A_{\alphas_t}(\f_1-\f_2).
	\end{align*}
	By the triangle inequality, we have
	\begin{align*}
	&\|\hat{l}^1_t(\f_1)-\hat{l}^1_t(\f_2)\|^2\leq(\|\f_2\|^2\|A_{\alphas_t}^\top A_{\alphas_t}\|^2+\|\shahat_t^\top A_{\alphas_t}\|^2\\&+\frac{1}{2}\|\f_1-\f_2\|^2\|A_{\alphas_t}^\top A_{\alphas_t}\|^2+\frac{\lambda_2}{2})\|\f_1-\f_2\|^2,
	\end{align*}
	where the norm is the operator norm if the argument is a matrix. Upon setting $$ L=\sqrt{2K\|A_{\alphas_t}^\top A_{\alphas_t}\|^2+\|\shahat_t^\top A_{\alphas_t}\|^2+\frac{\lambda_2}{2}}, $$
	and noticing that all terms are bounded independently from $t$, the local Lipschitz-continuity of $\hat{l}_t^1$ is obtained. The local Lipschitz-continuity of $\hat{l}_t$ is verified by observing that the function $g_2$ is clearly locally Lipschitz-continuous in the compact set $\setF$, and by noticing that the sum of locally Lipschitz-continuous functions is also locally Lipschitz-continuous.
	
	(iii) The sum of strongly convex functions is a strongly convex function, so we have to prove that $\hat{l}_t(\f)$ is strongly convex $\forall \shahat_t\in\chi$ to satisfy the statement. $\hat{l}_t^1$ is strongly convex since $(\forall t\in\Natural)(\forall \shahat_t\in\chi)$, its Hessian $\nabla^2\hat{l}_t^1(\f)\succeq \lambda_2 I$, i.e., the lowest eigenvalue of the matrix $\nabla^2\hat{l}_t^1(\f)-\lambda_2I$ is grater than zero. The sum of a convex function and a strongly convex function is a strongly convex function. To show this, let $\f_1,\f_2\in\setF$. Since $\hat{l}^1_t$ is strongly convex, we have
	$$ \hat{l}^1_t(\f_1)\geq\hat{l}^1_t(\f_2)+\nabla\hat{l}^1_t(\f_2)^\top(\f_1-\f_2)+\frac{c_2}{2}\|\f_1-\f_2\|^2. $$
	On the other hand, since $g_2$ is convex in $\setF$, then for some $\signal{g}\in\partial g_2(\f_2)$:
	$$ g_2(\f_1)\geq g_2(\f_2)+\signal{g}^\top(\f_1-\f_2).$$
	Adding both inequalities, we have that
	\begin{align*}
	&\hat{l}^1_t(\f_1)+g_2(\f_1)\geq \hat{l}^1_t(\f_2)+g_2(\f_2)+(\nabla\hat{l}_t(\f_2)\\&+\signal{g})^\top(\f_1-\f_2)+\frac{c_2}{2}\|\f_1-\f_2\|^2,
	\end{align*}
	so $\hat{l}_t$ is strongly convex with $\nabla\hat{l}_t(\f_2)+\signal{g}\subseteq\partial\hat{l}_t(\f_2).$

	(iv) $(\forall t\in\Natural)$ the boundedness of $l^1_t$, $\hat{l}^1_t$, $l^{1'}_t$, $\hat{l}^{1'}_t$, $\nabla^2l^1_t$ and $\nabla^2\hat{l}^1_t$ is automatically satisfied since the functions $l^1_t$ and $\hat{l}^1_t$ are continuously second order differentiable with respect to $\f\in\setF~\forall\shahat_t\in\chi$ and the set $\chi$ is bounded by Assumption 1 \cite{razaviyayn2016stochastic}. The boundedness of $g_2$, $g_2'$ and $\hat{l}_t$ follows from the Extreme value Theorem, since $g_2$, $g_2'$, and $(\forall\f\in\setF)~\hat{l}_t(\f)$ are continuous and the set $\setF$ is compact, and therefore they attain a maximum and a minimum in $\setF.$

\end{proof}

We now prove that, for consecutive estimates $\f_t$ and $\f_{t+1}$ of Alg.~\ref{alg:alt_minimization}, we have that $\|\f_t-\f_{t+1}\|\in\mathcal{O}(\frac{1}{t})$, a necessary condition for its convergence.

\begin{lemma}
	\label{lemma:decrease_f_t}
	Let $\f_t$ and $\f_{t+1}$ be the estimates of Alg.~\ref{alg:alt_minimization} after iterations $t$ and $t+1$, respectively. Under Assumptions 1 and 2, we have that 
	$$ \|\f_t-\f_{t+1}\|\in\mathcal{O}\left(\frac{1}{t}\right). $$
\end{lemma}
\begin{proof}
	
	From the definition of convex functions, we have $(\forall t\in\Natural)(\forall \f_1,\f_2 \in \setF)(\forall \signal{g}\in\partial \check{h}_t(\f_2))$: 

	$$\check{h}_t(\f_1)\geq \left \langle  \signal{g},\f_1-\f_2\right \rangle+\check{h}_t(\f_2).$$

	Note that $(\forall t\in\Natural)~\f_t$ is the minimizer of $\check{h}_t$ over the set $\setF$. Since $(\forall \f_2 \in \setF)~\check{h}_t(\f_2)-\check{h}_t(\f_t)\geq 0,$ we can write:
	$$\left\langle\signal{g},\f_2-\f_t\right \rangle\geq\check{h}_t(\f_2)-\check{h}_t(\f_t)\geq 0.$$
	From the strong convexity of $\check{h}_t$ we have that:
	\begin{equation}
	\label{eq:strong_convexity_h}
		\check{h}_t(\f_{t+1})-\check{h}_t(\f_{t})\geq \frac{c_2}{2}\|\f_{t+1}-\f_t\|^2.
	\end{equation}
	On the other hand,
	\begin{align*}
			&\check{h}_{t}(\f_{t+1})-\check{h}_t(\f_t)\\&=\check{h}_{t}(\f_{t+1})+\check{h}_{t+1}(\f_{t+1})-\check{h}_{t+1}(\f_{t+1})\\&+\check{h}_{t+1}(\f_{t})-\check{h}_{t+1}(\f_{t})-\check{h}_{t}(\f_{t})\\
			&\overset{(a)}{\leq}\check{h}_{t}(\f_{t+1})-\check{h}_{t+1}(\f_{t+1})+\check{h}_{t+1}(\f_{t})-\check{h}_{t}(\f_{t})\\
			&= \frac{1}{t}\sum_{\tau=1}^{t}\hat{l}_{\tau}(\f_{t+1})-\frac{1}{t+1}\sum_{\tau=1}^{t+1}\hat{l}_{\tau}(\f_{t+1})\\&+\frac{1}{t+1}\sum_{\tau=1}^{t+1}\hat{l}_{\tau}(\f_{t})-\frac{1}{t}\sum_{\tau=1}^{t}\hat{l}_{\tau}(\f_{t})\\
			&=\frac{1}{t(t+1)}\sum_{\tau=1}^{t}\hat{l}_{\tau}(\f_{t+1})-\frac{1}{t+1}\hat{l}_{t+1}(\f_{t+1})\\&-\frac{1}{t(t+1)}\sum_{\tau=1}^{t}\hat{l}_{\tau}(\f_{t})+\frac{1}{t+1}\hat{l}_{t+1}(\f_{t})\\
			&\leq\frac{1}{t(t+1)}\sum_{\tau=1}^{t}|\hat{l}_{\tau}(\f_{t+1})-\hat{l}_{\tau}(\f_{t})|\\&+\frac{1}{t+1}|\hat{l}_{t+1}(\f_{t+1})-\hat{l}_{t+1}(\f_{t})|
			\overset{(b)}{\leq}\frac{c_1}{t}\|\f_{t+1}-\f_{t}\|,
	\end{align*}
	where (a) follows from the fact that $\f_{t+1}$ is the minimizer of $\check{h}_{t+1}$, while (b) follows from the Lipschitz-continuity of $\hat{l}_t$ and $\hat{l}_{t+1}$. Combining \eqref{eq:strong_convexity_h} and (b) yields the desired result.
\end{proof}

The next result shows that the sequence of functions $(\check{h}_t)_{t=1}^\infty$ acts asymptotically as a surrogate of $h$. Moreover, we prove the almost sure convergence of $h$.

\begin{lemma}
	\label{lemma:convergence_objective}
	Let $(\forall t\in\Natural)~h_t:\setF\rightarrow\real$ be the function defined in \eqref{eq:h_c_f_b_equals_h_f_b}, let $h$ be its limit when $t\rightarrow\infty$ as in \eqref{eq:h}, and denote by $\check{h}_t:\setF\rightarrow\real$ the surrogate function defined in \eqref{eq:surrogate}. Under Assumptions 1 and 2, the following is true: 
	\begin{enumerate}[label=(\roman*)]
		\item $(\check{h}_t(\f_t))_{t=1}^\infty$ converges a.s.,
		\item $ \underset{t\rightarrow\infty}{\textup{lim}}\!~\left(\check{h}_t(\f_t)-h_t(\f_t)\right)=0$ a.s.,
		\item $ \underset{t\rightarrow\infty}{\textup{lim}}\check{h}_t(\f_t)-h(\f_t)=0$ a.s., and
		\item $(h(\f_t))_{t=1}^\infty$ converges a.s.
	\end{enumerate}
\end{lemma}
\begin{proof}
	The proof requires the use of the convergence of empirical processes \cite{van2000asymptotic} and of quasi-martingales \cite{fisk1965quasi}. First, we have 
		\begin{multline}
		\check{h}_{t+1}(\f_{t+1})-\check{h}_{t}(\f_{t})\\=\check{h}_{t+1}(\f_{t+1})-\check{h}_{t}(\f_{t})+\check{h}_{t+1}(\f_{t})-\check{h}_{t+1}(\f_{t})\\
		\leq\frac{l_{t+1}(\f_t)-\check{h}_t(\f_t)}{t+1} \label{subeq:bound_u_t}
		\end{multline}
	after noticing that $ \check{h}_{t+1}(\f_{t})=\frac{l_{t+1}(\f_t)+\check{h}_t(\f_t)}{t+1}, $ 
	$\check{h}_{t+1}(\f_{t+1})-\check{h}_{t+1}(\f_t)\leq 0$ since $\f_{t+1}$ minimizes $\check{h}_{t+1}$, and $h_t(\f_t)-\check{h}_t(\f_t)\leq 0$, since $\check{h}_t$ upperbounds the empirical cost $h_t$. 
	
	Let $\mathcal{N}_t$ be the filtration of the past information of $\check{h}_t$, and let $\sigma(X_1,...,X_t)$ be the $\sigma$-algebra generated by the random variables $X_1,...,X_t$. The filtration of $\check{h}_t$ up to $t$ is defined as $ \mathcal{N}_t=\sigma(\check{h}_1,...,\check{h}_t) $. Taking the expectation with respect to the filtration $\mathcal{N}_t$, one can write
	\begin{align*}
		&\mathbb{E}\left [ \check{h}_{t+1}(\f_{t+1})-\check{h}_t(\f_{t})|\mathcal{N}_t \right ] \leq\E\left[\frac{l_{t+1}(\f_t)-\check{h}_t(\f_t)}{t+1}|\setN_t\right]  \\
		&=\frac{h(\f_t)}{t+1}-\frac{\check{h}_t(\f_t)}{t+1}\\
		&=\frac{h(\f_t)-h_t(\f_t)}{t+1}-\frac{\check{h}_t(\f_t)-h_t(\f_t)}{t+1}\\
		&\overset{(a)}{\leq}\frac{h(\f_t)-h_t(\f_t)}{t+1} \\
		&\overset{(b)}{\leq}\frac{\|h-h_t\|_\infty}{t+1},
	\end{align*}
	where (a) is obtained from the fact that $\check{h}_t(\f_t)\geq h_t(\f_t)$ and (b) from the definition of $\|\cdot\|_\infty$.
	
	We need now to prove that $\|h-h_t\|_\infty<\infty$. To this end, we can make use of Fact \ref{lemma:donsker}, a corollary of the Donsker's theorem, which states that, under some necessary conditions, $\E[\sqrt{t}\|h-h_t\|_\infty]<\infty$. Concretely, we can verify that all necessary conditions are fulfilled in our case, namely, that $(\forall t\in\Natural)~\ell_t$ is Lipschitz continuous and bounded by Lemma \ref{lemma:regularity}, that the set $\setF$ is bounded by Assumption \ref{ass:compact_set_f_it}, and that $\E_\shahat[\ell_t(\f)^2]$ exists and is uniformly bounded. 
	
	The Donsker's theorem also implies that there exists a constant $K_1>0$ such that
	\begin{equation}
	\label{eq:quasi_martingale}
		\mathbb{E}\left[\mathbb{E}\left [ \check{h}_{t+1}(\f_{t+1})-\check{h}_t(\f_{t})|\mathcal{N}_t \right ]^+\right]\leq\frac{K_1}{t^{3/2}},
	\end{equation}
	where the operator $[\cdot]^+$ represents the projection onto the non-negative orthant. Summing \eqref{eq:quasi_martingale} over $t,$ we obtain
	$$ \sum_{t=1}^{\infty}\mathbb{E}\left[\mathbb{E}\left [ \check{h}_{t+1}(\f_{t+1})-\check{h}_t(\f_{t})|\mathcal{N}_t \right ]^+\right]<\infty. $$
	
	Now by applying Fact \ref{theorem:quasi_martingales} on the convergence of quasi-martingales, we obtain both that
	\begin{equation}
		\label{eq:convergence_quasi_martingale}
		\sum_{t=1}^{\infty}\left|\mathbb{E}\left[ \check{h}_{t+1}(\f_{t+1})-\check{h}_t(\f_{t})|\mathcal{N}_t \right]\right|<\infty~ a.s.,
	\end{equation}
	and that $\check{h}_t(\f_t)$ converges almost surely, obtaining (i).

	Using \eqref{subeq:bound_u_t} and \eqref{eq:convergence_quasi_martingale}, it can be shown \cite[Lemma 1]{razaviyayn2016stochastic} the almost sure convergence of the positive sum
	\begin{equation}
		\label{eq:convergence_sum_h_chech_t_minus_h_t}
		\sum_{t=1}^{\infty}\frac{\check{h}_t(\f_t)-h_t(\f_t)}{t+1}<\infty~a.s.
	\end{equation}
	Using Lemma \ref{lemma:decrease_f_t}, Eq.~\eqref{eq:convergence_sum_h_chech_t_minus_h_t}, and the fact that $\check{h}_t(\f_t)-h_t(\f_t)\geq 0$, the hypotheses of Fact \ref{lemma:positive_congerging_sum} on positive converging sums can be verified and Fact \ref{lemma:positive_congerging_sum} can be applied, obtaining (ii).
	
	In addition, we can use the Glivenko-Cantelli theorem (Fact \ref{theorem:Glivenko}) that determines the asymptotic behavior of the empirical distribution function as the number of i.i.d. observations grows, which gives us $$\underset{t\rightarrow \infty}{\textup{lim}}~\|h_t-h\|_\infty=0~a.s.$$ Therefore,
	$$ \underset{t\rightarrow\infty}{\textup{lim}}h(\f_t)-\check{h}_t(\f_t)\rightarrow 0~a.s., $$
	and $(h(\f_t))_{t=1}^\infty$ converges almost surely, which proves (iii) and (iv).
	
	
\end{proof}

We are now in position of proving the convergence of Alg.~\ref{alg:alt_minimization} to a stationary point of Problem \eqref{eq:min_empirical_cost}.

\begin{proposition}
	\label{theorem:convergence_args}
	Let $(\f_t)_{t=1}^\infty$ be the sequence of iterates generated by Alg.~\ref{alg:alt_minimization} with parameters $\lambda_1,\lambda_2,\lambda_3>0$ when started with an arbitrary $\f_0\in\setF$. Suppose that Assumptions 1 and 2 are satisfied. Then, the following statements are true:
	\begin{enumerate}[label=(\roman*)]
		\item Let $(\forall t\in\Natural)~h^1_t:\setF\rightarrow\real, $ and $l^1_t:\setF\rightarrow\real$ be the functions defined in \eqref{eq:h_1_t} and \eqref{eq:l_1_t}, respectively. The limit function of $h^1_t$ when $t\rightarrow\infty$, defined by 
		$h_1(\f)=\underset{t\rightarrow\infty}{\textup{lim}}h^1_t(\f)=\mathbb{E}_\shahat[l^1_t(\f)],$
		 exists.
		\item Let $(\forall t\in\Natural)~\check{h}^1_t:\setF\rightarrow\real$ be the function defined in \eqref{eq:h_check_1_t}. The sequences of functions $(h^1_{t})_{t=1}^\infty$ and $(\check{h}^1_{t})_{t=1}^\infty$ are equicontinuous.
		\item Let $(\f_{t_j})_{j=1}^\infty$ be a subsequence converging to a point $\f^*\in\setF.$ Then,
			\begin{equation}
		\label{eq:h_c_f_b_equals_h_f_b}
		\check{h}^1(\f^*)=h^1(\f^*).
		\end{equation}
		\item Let $\setF^*$ denote the set of stationary points of the \ac{SLF} Problem \eqref{eq:min_empirical_cost}. Suppose that $\f^*\in\mathrm{int}(\setF),$ then,
		$$ \underset{t\rightarrow\infty}{\textup{lim}}~\underset{\f\in\setF^*}{\textup{inf}}\|\f_t-\f\|=0~a.s.$$
	\end{enumerate}
	
\end{proposition}
\begin{proof}
	The proof uses the Arzelà–Ascoli theorem (Fact \ref{theorem:arzela}), which gives necessary and sufficient conditions to decide whether every sequence of a given family of real-valued continuous functions defined on a compact set has a uniformly convergent subsequence.
	
	(i) From Lemma \ref{lemma:regularity} (iv) we know that $(\forall (t,\f)\in\Natural\times\setF)(\forall \shahat_t\in\chi)~\|l^1_t(\f)\|\leq K_3,$ where $K_3>0$ is a constant. Then, the limit of $h^1_t(\f)$ when $t\rightarrow\infty$ exists and the proof follows from the strong law of large numbers \cite{fristedt2013modern}.
	
	(ii) Since $(\forall (t,\f)\in(\Natural\times\setF))(\forall \shahat_t\in\chi)~\|\nabla l^1_t(\f)\|\leq K_3$ by Lemma \ref{lemma:regularity} (iv), by the mean value theorem we have that the sequence of functions $(h^1_{t})_{t=1}^\infty$ is equicontinuous. Also from Lemma \ref{lemma:regularity} (iv), we have that $\|\nabla \hat{l}^1_t(\f)\|\leq K_3$. Then, the sequence of functions $(\check{h}^1_{t})_{t=1}^\infty$ is also equicontinuous, bounded and defined over the compact set $\setF$. 
	
	(iii) Consider a subsequence $(h^1_{t_j})_{j=1}^\infty$. By restricting to this subsequence, we have 
	\begin{equation}
	\label{eq:e_l_t_1_f_bar}
	\underset{j\rightarrow\infty}{\textup{lim}}h^1_{t_j}(\f_{t_j})=\mathbb{E}_\shahat[l^1_t(\f^*)].
	\end{equation}
	Applying the Arzelà–Ascoli theorem implies that, by restricting to a subsequence $(\check{h}^1_{t_j})_{j=1}^\infty$, there exists a uniformly continuous function $\check{h}^1(\f)$ such that 
	\begin{equation}
	\label{eq:lim_h1_tj}
	(\forall\f\in\setF)~\underset{j\rightarrow\infty}{\textup{lim}}\check{h}^1_{t_j}(\f)=\check{h}^1(\f),
	\end{equation}
	and therefore
	\begin{equation}
	\label{eq:lim_h1_tj_f_tj}
	\underset{j\rightarrow\infty}{\textup{lim}}\check{h}^1_{t_j}(\f_{t_j})=\check{h}^1(\f^*).
	\end{equation}	
	We know by definition that $(\forall \f\in\setF)~\check{h}^1_{t_j}(\f)\geq h_{t_j}^1(\f).$ 
	Letting $j\rightarrow\infty$, we obtain
	\begin{equation}
	\label{eq:h_check_geq_h}
	(\forall\f\in\setF)~\check{h}^1(\f)\geq h^1(\f).
	\end{equation}
	Using Lemma \ref{lemma:convergence_objective} (ii), \eqref{eq:e_l_t_1_f_bar} and \eqref{eq:lim_h1_tj_f_tj}, Eq.~\eqref{eq:h_c_f_b_equals_h_f_b} yields.

	(iv) Define the function 
	$u:\setF\rightarrow\real: \f\mapsto \check{h}^1(\f)-h^1(\f).$
	From \eqref{eq:h_check_geq_h}, we know that $(\forall\f\in\setF)~u(\f)\geq0,$ and $u$ attains a minimum at $\f^*$ due to \eqref{eq:h_c_f_b_equals_h_f_b}. Since $\f^*\in\mathrm{int}(\setF),$ 
	the first order optimality condition in $u(\f^*)$ implies that $\nabla u(\f^*)=\nabla\check{h}^1(\f^*)-\nabla h^1(\f^*)=0,$ or equivalently 
	\begin{equation}
		\label{eq:nabla_eq_nabla}
		\nabla\check{h}^1(\f^*)=\nabla h^1(\f^*).
	\end{equation} 
	Using the updates of Alg.~\ref{alg:alt_minimization}, we have 
	$$ (\forall\f\in\setF)~\check{h}_{t_j}(\f_{t_j})=\check{h}^1_{t_j}(\f_{t_j})+g_2(\f_{t_j})\leq\check{h}^1_{t_j}(\f)+g_2(\f). $$
	Letting $j\rightarrow\infty$, and using \eqref{eq:lim_h1_tj_f_tj} and the fact that the sequence $(\check{h}^1_{t_j})_{j=1}^\infty$ is equicontinuous, it yields
	\begin{equation}
		\label{eq:h_f_bar_leq_h}
		(\forall\f\in\setF)~\check{h}^1(\f^*)+g_2(\f^*)\leq\check{h}^1(\f)+g_2(\f).
	\end{equation}
	We have to show the existence of the directional derivative $(\forall\di\in\real^P)~\check{h}'(\f;\di),$ which is not guaranteed because the set $\setF$ is compact by Assumption \ref{ass:compact_set_f_it}. To obtain its existence, let us first define a new function 
	$$\tilde{h}:\real^P\rightarrow\real:~\f\mapsto \underset{t\rightarrow\infty}{\limit}~\check{h}_t(\f)~a.s,$$ 
	which, provided that the limit exists, it is a continuous and convex function. From Fact \ref{prop:exitence_direc_deriv}, we have that for any $\tilde{\f}\in\mathrm{int}(\mathrm{dom}~\tilde{h})$, the directional derivative $(\forall \di\in\real^P)~\tilde{h}'(\tilde{\f};\di)$ is a real number. Note that, since $\setF\subseteq \mathrm{int}(\mathrm{dom}~\tilde{h}),$ the function $\check{h}:\setF\rightarrow\real$ is the restriction of the function $\tilde{h}$ over the compact set $\setF$. Therefore, we conclude that $(\forall\f\in\setF)(\forall\di\in\real^P)~\check{h}'(\f;\di)$ exists. Equation \eqref{eq:h_f_bar_leq_h} implies that $\f^*$ is a minimizer of $\check{h}$, and, combining this fact with the existence of its directional derivative, we have
	$$ (\forall\di\in\real^P)~\left \langle \nabla\check{h}^1(\f^*),\di \right \rangle +g_2'(\f^*;\di)\geq 0.$$
	Combining this result with \eqref{eq:nabla_eq_nabla} and given the fact that $h=h_1+g_2$, we obtain
	$$ (\forall\di\in\real^P)~h'(\f^*;\di)\geq0, $$
	which means that $\f^*$ is a stationary point of $h$.
	
	
\end{proof}

The remaining step in the convergence analysis consists in showing that the online algorithm also converges in the arguments to a stationary point of Problem \eqref{eq:min_empirical_cost}, which represents our most important result. We formally state this fact in the following theorem.

\begin{theorem}
	\label{theorem:convergence_args_alg1}
	Let $(\f_t)_{t=1}^\infty$ be the sequence of iterates generated by Alg.~\ref{alg:online} with parameters $\lambda_1,\lambda_2,\lambda_3>0$ when started with an arbitrary $\f^{(0)}_1\in\setF$. Suppose that Assumptions 1 and 2 are satisfied. Then, the following statements are true:
	\begin{enumerate}[label=(\roman*)]
		\item $(\forall t\in\Natural)(\forall n\in\Natural)$ let $L_g$ be the Lipschitz-constant of the gradient of $g_1$ defined in \eqref{eq:g_1}, and let $L_k$ be the Lipschitz-constant of the gradient of $k_t$ defined in \eqref{eq:least_squares_problem}. Define $T_1:\setF\rightarrow\setF:\f\mapsto \f-\gamma\nabla_{g_1}(\f)$ and $T_2:\setF\rightarrow\setF:\f\mapsto \textup{soft}_{\lambda_1}(\f).$ If $0<\gamma\leq(1-\epsilon)L_g$ and $0<\mu\leq(1-\epsilon)L_k$, for some $\epsilon>0$, then the composition mapping $T:\setF\rightarrow\setF:\f\mapsto T_2T_1(\f)$ of $T_2$ and $T_1$ is a contraction mapping.
		\item $ (\forall t\in\Natural) $ let $T^{(n)}$ be the mapping $T$ at iteration index $n\in\Natural.$ Define the mapping $F_n:\setF\rightarrow\setF:\f\mapsto T^{(n)}T^{(n-1)}...T^{(1)}(\f)$ as the composition of $T^{(1)},T^{(2)},...,T^{(n)}$. Then, the range of the mapping $F_\infty$ defined as $F_\infty=\underset{n\rightarrow\infty}{\limit}~F_n(\f)$ is a singleton.
		\item $ (\forall t\in\Natural) $ let $(\f_t^{(n)})_{n=1}^\infty$ be the sequence generated by the inner loop of Alg.\ref{alg:online}. Then, $(\f_t^{(n)})_{n=1}^\infty\rightarrow\f_t\in\setF.$
		\item $(\forall t\in\Natural)$ suppose that the stopping criterion of Alg.~\ref{alg:online} is reached after $N\in\Natural$ iterations. Let $(\alphas_1^{(N)},...,\alphas_t^{(N)},\f_t)\in\setC_1^\alphas\times...\times\setC_t^\alphas\times\setF$ be the estimates at time $t$ of Alg.~\ref{alg:online}. Define $\check{h}_t:\setF\rightarrow\real$ as in \eqref{eq:surrogate} with $\alphas_1=\alphas_1^{(N)},...,\alphas_t=\alphas_t^{(N)}.$ Let $(\f_{t_j})_{j=1}^\infty$ be a subsequence converging to a point $\f^*\in\setF,$ and let $\setF^*$ be the set of stationary points of the \ac{SLF} Problem \eqref{eq:min_empirical_cost}. Suppose that $\f^*\in\mathrm{int}(\setF).$ Then,
		$$ \underset{t\rightarrow\infty}{\textup{lim}}~\underset{\f\in\setF^*}{\textup{inf}}\|\f_t-\f\|=0~a.s.$$
	\end{enumerate}
\end{theorem}
\begin{proof}
	For this proof we exploit concepts of contraction and non-expansive mappings. The outline of the proof is as follows: first, we show that each iteration of Alg.~\ref{alg:online} is the composition of a (firmly) non-expansive mapping and a contraction mapping, which is also a contraction mapping. Second, we show that the iterations of Alg.~\ref{alg:online} can be seen as the composition of infinitely many contraction mappings, and such a mapping sends each point in $\setF$ onto its unique representation $\{\f_t\}$. In other words, the range of the mapping is a singleton. With this we obtain the convergence of Alg.~\ref{alg:online} in the arguments. The final step of the proof consists in showing that $\f_t$ belongs to the set of stationary points $\setF^*$ of Problem \eqref{eq:min_empirical_cost} when $t\rightarrow\infty$.
	
	(i) Note that the $nth$ SLF-iteration of Alg.~\ref{alg:online} in Equation \eqref{eq:its_f_problem} can be expressed as the evaluation of $T$ at $\f^{(n-1)}.$ Also note that $T_1$ represents the famous gradient descent algorithm. Given the fact that $g_1$ is strongly convex and Lipschitz-differentiable with its gradient's Lipschitz constant $L_g=\|A_{\alpha^{(n)}}^\top A_{\alpha^{(n)}}+\lambda_2\|_2$ (here $\|\cdot\|_2$ is the spectral norm of a matrix), and that $\gamma\leq(1-\epsilon)/L_g$ due to algorithmic design with $\epsilon>0$, it can be shown that $T_1$ is a contraction mapping with contraction factor $c\in[0,1-\epsilon]$, as follows: let $\f_1,~\f_2\in\setF,$ then
	\begin{align*}
		\|T_1(\f_1)-T_1(\f_2)\|&=\|\f_1-\gamma\nabla_{g_1}(\f_1)-\f_2+\gamma\nabla_{g_1}(\f_2)\|\\
		&\overset{(a)}{=}\|\f_1-\f_2-\gamma\nabla^2 g_1(\signal{z})(\f_1-\f_2)\|\\
		&=\|(\f_1-\f_2)(I-\gamma\nabla^2 g_1(\signal{z}))\|\\
		&\overset{(b)}{\leq}\|\f_1-\f_2\|(1-L_g\gamma)=c\|\f_1-\f_2\|,		
	\end{align*}
	where $\signal{z}=\eta\f_1+(1-\eta)\f_2$ for some $\eta\in[0,1]$. We used in step (a) the mean value theorem of vector calculus (Fact \ref{theo:mean_value}), and step (b) stems from the fact that $\nabla^2 g_1\succeq L_g I.$
	
	Note that due to Assumption \ref{ass:iid}, the iterations in \eqref{eq:projected_gradient_its} are well behaved in the sense that $(\forall n\in\Natural)(\forall t\in\Natural)~\alphas_t^{(n)}$ belongs to the compact set $\setC_t^{\alpha}$. Following similar arguments as with $L_g$, we obtain the range of $\mu\in[0,(1-\epsilon)/L_k]$, with $L_k$ being the Lipschitz constant of the gradient of $k_t$ from Eq.~\eqref{eq:least_squares_problem} given by $L_k=\|A_\f^\top A_\f+\lambda_3I\|_2$, where $\|\cdot\|_2$ here is the spectral norm of a matrix.
	
	On the other hand, we know from Fact \ref{lemma:non_expansiveness} that the proximal operator of a lower semi-continuous function is a (firmly) non-expansive mapping, given in our case by $T_2$. Let $\f_1'=T_1(\f_1)$ and $\f_2'=T_1(\f_2)$. Note that, from the definition of firmly non-expansive mappings, we have
	\begin{align*}
		\|T_2(\f'_1)-T_2(\f'_2)\|^2&\leq\|T_2(\f'_1)-T_2(\f'_2)\|^2\\&+\|(\Id-T_2)(\f'_1)-(\Id-T_2)(\f'_2)\|^2\\&\leq\|\f'_1-\f'_2\|^2.
	\end{align*}
	Thus, we can write
	\begin{align*}
		\|T_2T_1(\f_1)-T_2T_1(\f_2)\|&=\|T(\f_1)-T(\f_2)\|\\&\leq\|T_1(\f_1)-T_1(\f_2)\|\\&\leq c\|\f_1-\f_2\|,
	\end{align*}
	which follows from the definition of $T$ and from the fact that $T_1$ is a contraction mapping.
	
	(ii) By induction, we obtain that $F_n$ is also a contraction mapping, since
	\begin{equation*}
		\|F_n(\f_1)-F_n(\f_2)\|\leq C\|\f_1-\f_2\|,
	\end{equation*}
	with $C=\prod_{j=1}^{n}c_j\in]0,(1-\epsilon)^n]$. Now, since $F_\infty$ is the composition of $T^{(1)},T^{(2)},...,T^{(n)}$ when $n\rightarrow\infty$, we obtain the following:
	\begin{equation*}
		\|F_\infty(\f_1)-F_\infty(\f_2)\|\leq\underset{n\rightarrow\infty}{\limit}~\prod_{j=1}^{n}c_j ~\|\f_1-\f_2\|=0,
	\end{equation*}
	which means that $F_\infty(\f_1)=F_\infty(\f_2)\forall\f_1,\f_2\in\setF$, and therefore the range of $F_\infty$ is a singleton. 
	
	(iii) Let $\{\f_t\},~\f_t\in\setF$, be the range of $F_\infty$. The sequence generated by $F_\infty$ is a constant sequence given by $(\f_t)_{n=1}^\infty$, and therefore $\|\f^{(n)}-\f_t\|=0$, which proves that $(\forall t\in\Natural)~(\f^{(n)}_t)_{n=1}^\infty\rightarrow \f_t\in\setF$. 
	
	(iv) Consider the function 
	\begin{align*}
		u_t(\alphas_1,...,\alphas_t,\f)&\triangleq\frac{1}{t}\sum_{\tau=1}^t\frac{1}{2}\|\shahat_{\tau}-A_\f K_\tau\alphas_\tau\|^2_2\\&+\lambda_1\|\f\|_1+\frac{1}{2}\lambda_2\|\f\|_2^2.
	\end{align*}
	It is easy to see that $(\alphas_1^{(N)},...,\alphas_t^{(N)},\f_t)$ is a bistable point (Definition \ref{def:bistable_point}) of $u_t$, since $u_t(\alphas_1^{(N)},...,\alphas_t^{(N)},\f_t)$ $\leq u_t(\alphas_1^{(N)},...,\alphas_t^{(N)},\f)~\forall\f\in\setF,$ and $u_t(\alphas_1^{(N)},...,\alphas_t^{(N)},\f_t)\leq u_t(\alphas_1,...,\alphas_t,\f_t)~\forall(\alphas_1,...,\alphas_t)\in\setC_1^\alpha\times...\times\setC_t^\alpha$ due to monotonicity\footnote{The monotonicity of the projected gradient descent is ensured since $k_t$ is strongly convex \cite{bubeck2014convex}.} in the iterates of the inner loop of Alg.\ref{alg:online}. Now consider the problem
	\begin{equation}
		\label{eq:problem_l_theo_2}
		\underset{\f}{\textup{minimize}}~\frac{1}{t}\sum_{\tau=1}^t\frac{1}{2}\|\shahat_{\tau}-A_\f K_\tau\alphas_\tau^{(N)}\|^2_2+\lambda_1\|\f\|_1+\frac{1}{2}\lambda_2\|\f\|_2^2.
	\end{equation}
	Since $(\alphas_1^{(N)},...,\alphas_t^{(N)},\f_t)$ is a bistable point of $u_t,$ $\f_t$ is the unique solution to Problem \eqref{eq:problem_l_theo_2}. Then, $\f_t$ is also the solution to the problem of minimizing the surrogate function $\check{h}_t(\f)$ with $\alphas_1=\alphas_1^{(N)},...,\alphas_t=\alphas_t^{(N)}$. We can use Lemma \ref{lemma:convergence_objective} to guarantee the almost sure convergence of $(\check{h}_t(\f))_{t=1}^\infty$ with Alg.~\ref{alg:online},
	and the same reasoning as in Proposition \ref{theorem:convergence_args} holds for Alg.~\ref{alg:online}, thus $(\f_t)_{t=1}^\infty\rightarrow\f^*\in\setF^*.$
	
\end{proof}

\section{Numerical Evaluation}
\label{sec:numerical_evaluation}
This section is devoted to the numerical evaluation of our proposed online algorithm for \ac{A2A} path loss maps learning. To this end, we consider two scenarios based on \ac{V2V} communications. The first scenario is based on synthetic data generated from the well-known Madrid scenario \cite{metis_D_6_1}. With the second scenario, we show the algorithm performance with more realistic data generated with the \ac{GEMV2} software \cite{boban2014geometry}. It has been shown \cite{boban2014geometry} that the \ac{GEMV2} model generates accurate path loss datasets for \ac{V2V} communications, so it is a good proxy for the evaluation of our algorithm with a realistic path loss dataset.

\subsection{Evaluation with Synthetic Data}
\label{sec:num_ev_synthetic}

The Madrid scenario \cite{metis_D_6_1} is plotted in Fig.~\ref{fig:madrid_scenario}. The original scenario has a size of $140\times 97$ meters, and we discretize it into a $56\times 39$ map, with each pixel being $2.5\times 2.5$ meters of size, since it roughly represents the size of a car. The map has seven $13\times 13$ buildings, one $13\times 13$ park, and other eight $13 \times 3$ buildings. The rest of the scenario represents roads connecting the different parts of the map. The normalized \ac{SLF} at each location, i.e. the attenuation that a link experiences while crossing that location, is set for buildings at 1, for the park at 0.1, and for road pixels at 0, since the \ac{SLF} of the air is considered to be negligible. Figure \ref{fig:slf_map} shows the yielding normalized \ac{SLF}. Vehicles are only allowed to be at road locations, which means that no measurements inside the buildings and park are acquired. This poses a major challenge to the algorithm, since there are many grouped locations for which no measurements can be acquired. Still, as we will see in this section, our algorithm is able to reconstruct the structures with high accuracy.

To generate a synthetic window function, we use the normalized elliptical model from \cite{hamilton2013propagation} and reproduced in \eqref{eq:elliptical_model}. We set the wavelength to $\eta = 0.1499$m in our simulations. The maximum number of vehicles, which coincides with the total number of road locations, is $P_{\textup{tx}}=744$. The total number of links in the map is given by $T=P(P-1)/2=2383836$, with $P=56\times 39=2184$ total pixels. Because not all locations in the map can be occupied by vehicles, the samples acquired are drawn from a subset of all possible links with cardinality $T_{\textup{tx}}=744(744-1)/2=276396$. This means that the maximum percentage of samples available is $11.59\%$ of the total. We assume that the samples arrive sequentially in $t=1,..,t_{\textup{max}}$, with $t_{\textup{max}}=200$ time steps. At each time instant $t$, $M=120$ i.i.d. new samples are acquired and one outer iteration of Alg.~\ref{alg:online} is run. With this setup, the total number of acquired measurements is $Mt_{\textup{max}}=24000$, which represents at most an $8.68\%$ of $T_{\textup{tx}}$, and, in turn, a mere $1.01\%$ of $T$, this is, of all possible links in the map. Other simulation parameters are $\sigma= 0.0001$, $\lambda_1=0.0004$, $\lambda_2=0.00001$, and $\lambda_3=0.00022$. Table \ref{table:simulation_parameters_2} summarizes the main simulation parameters.
\begin{figure}
	\centering
	
	\begin{subfigure}[b]{.24\textwidth}
		\centering
		\includegraphics[width=1\textwidth]{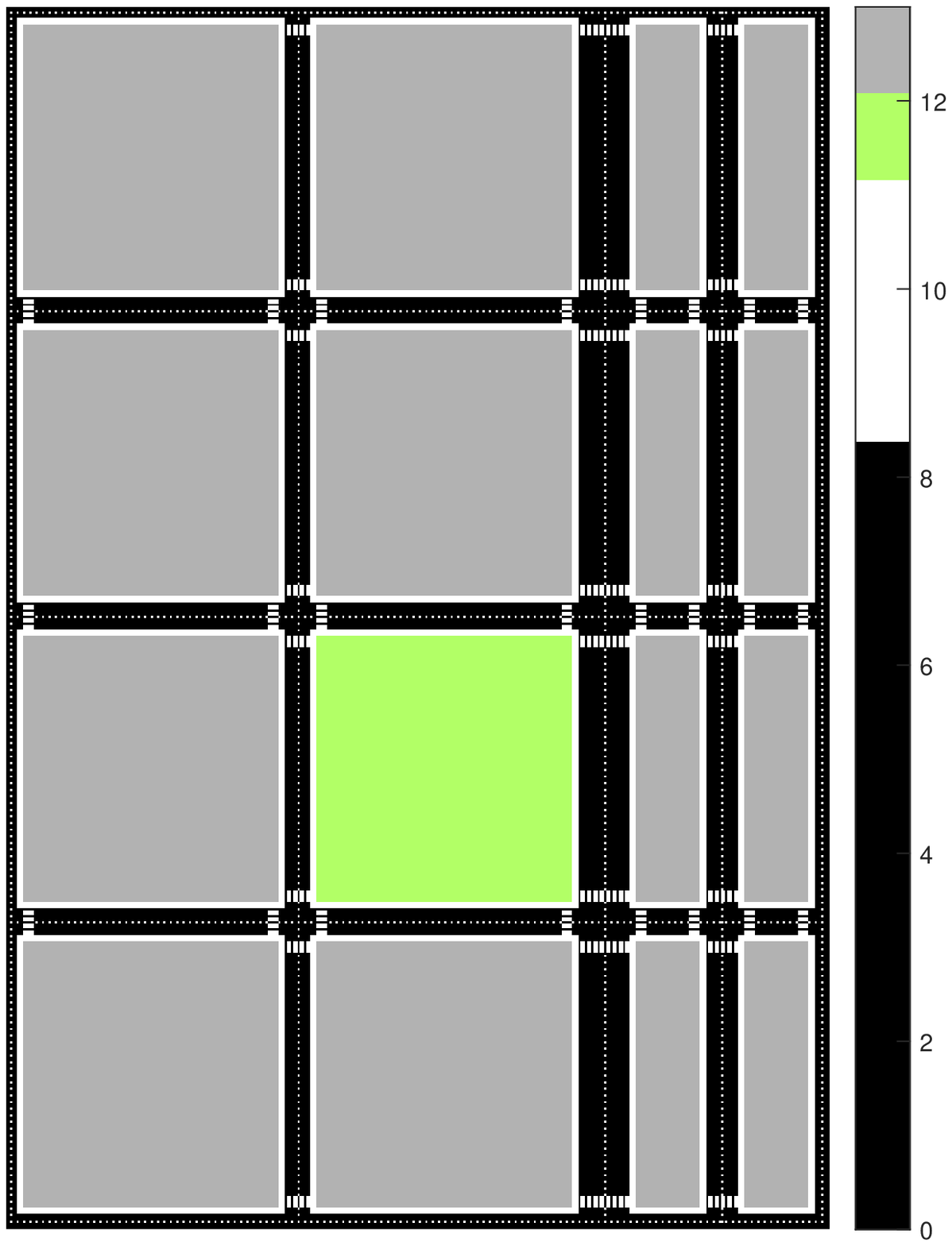}
		\caption{Original Madrid scenario.}
		\label{fig:madrid_scenario}
	\end{subfigure}
	\hfill
	\begin{subfigure}[b]{.24\textwidth}
		\centering
		\includegraphics[width=1\textwidth]{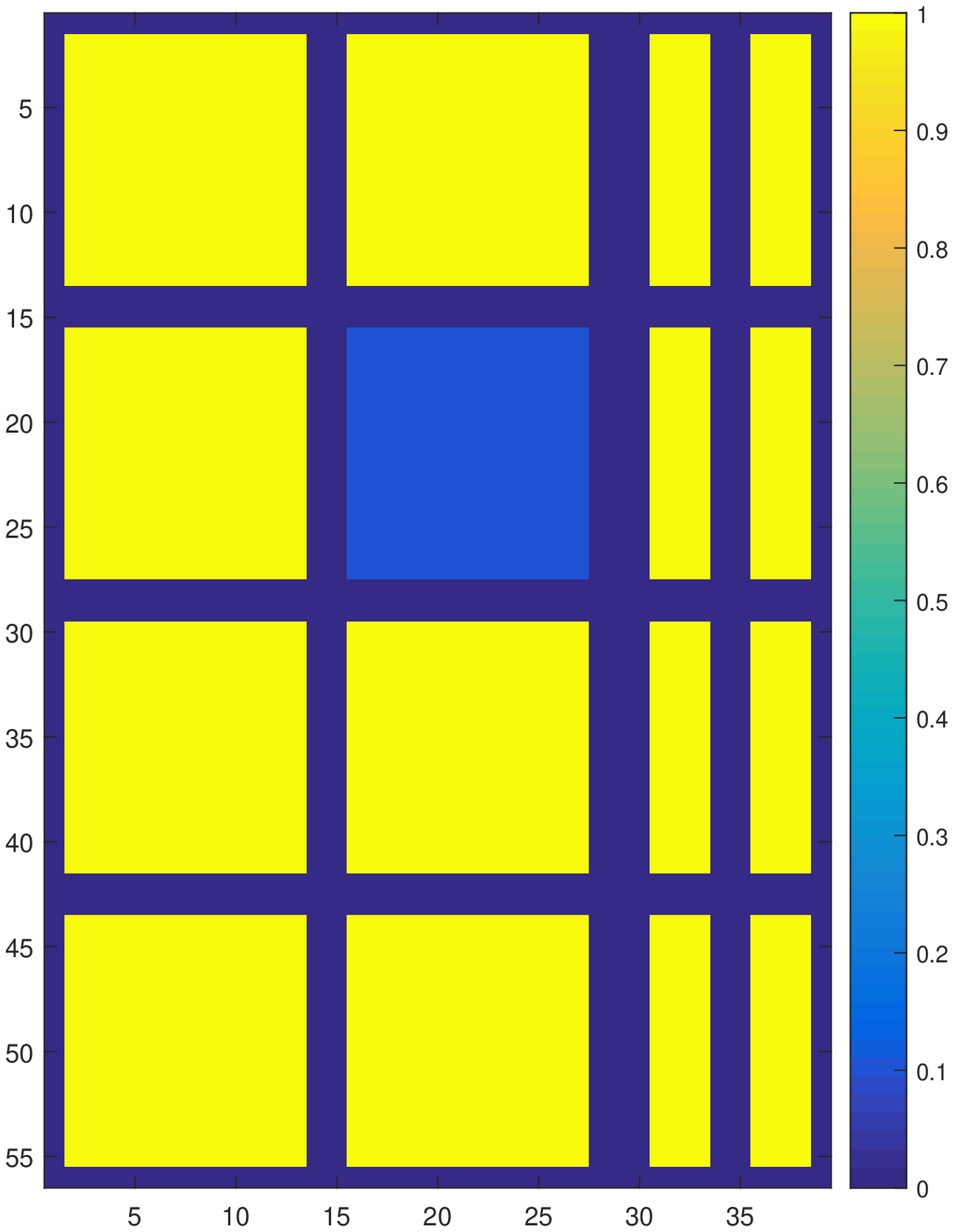}
		\caption{Normalized \ac{SLF} map.}
		\label{fig:slf_map}
	\end{subfigure}
	\label{fig:maps}
	\caption{Madrid scenario (left) and its normalized \ac{SLF} map (right).}
\end{figure}
\begin{figure}
	\centering
	
	\begin{subfigure}[b]{.24\textwidth}
		\centering
		\includegraphics[width=1\textwidth]{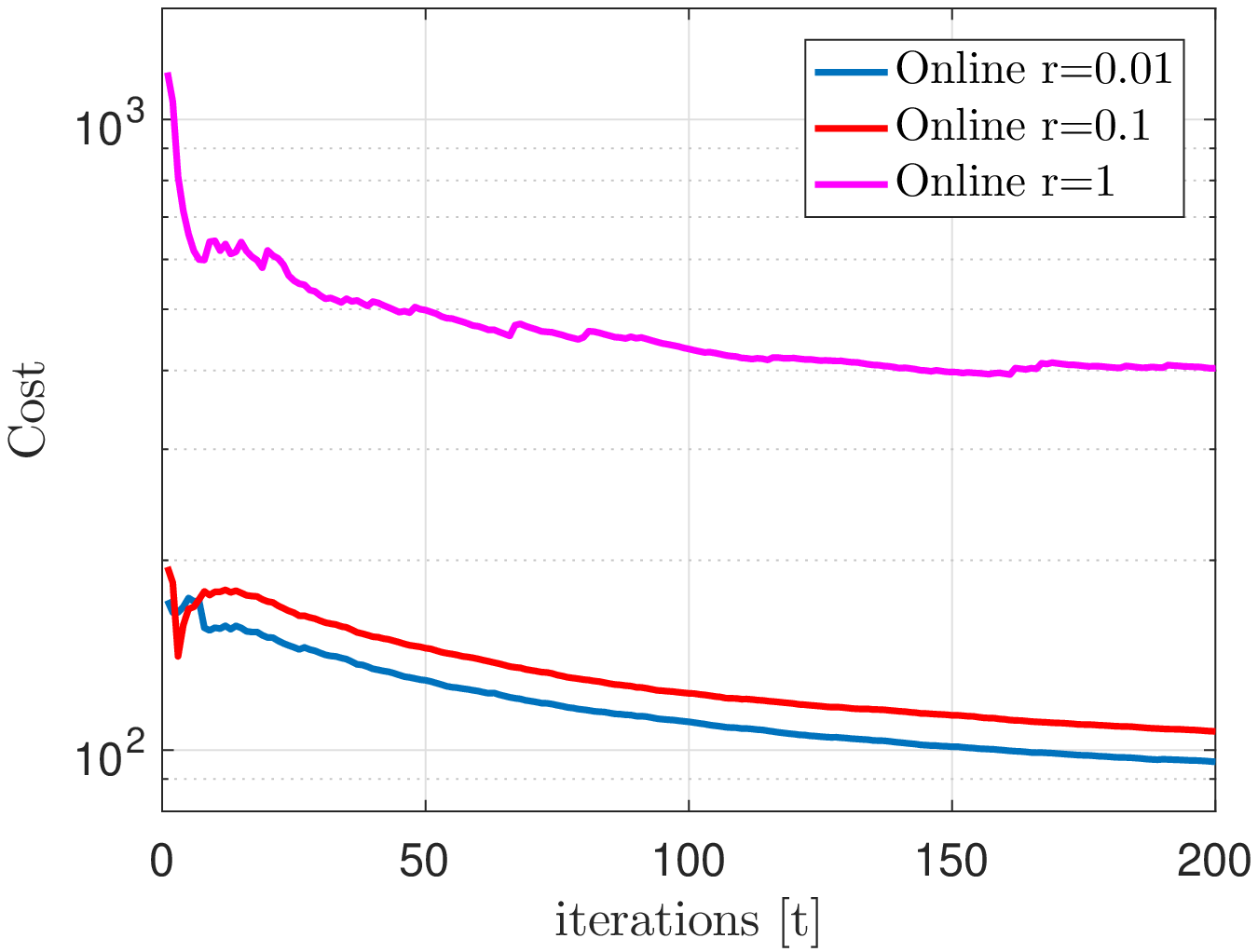}
		\caption{Cost versus $t$.}
		\label{fig:o_vs_it}
	\end{subfigure}
	\hfill
	\begin{subfigure}[b]{.24\textwidth}
		\centering
		\includegraphics[width=1\textwidth]{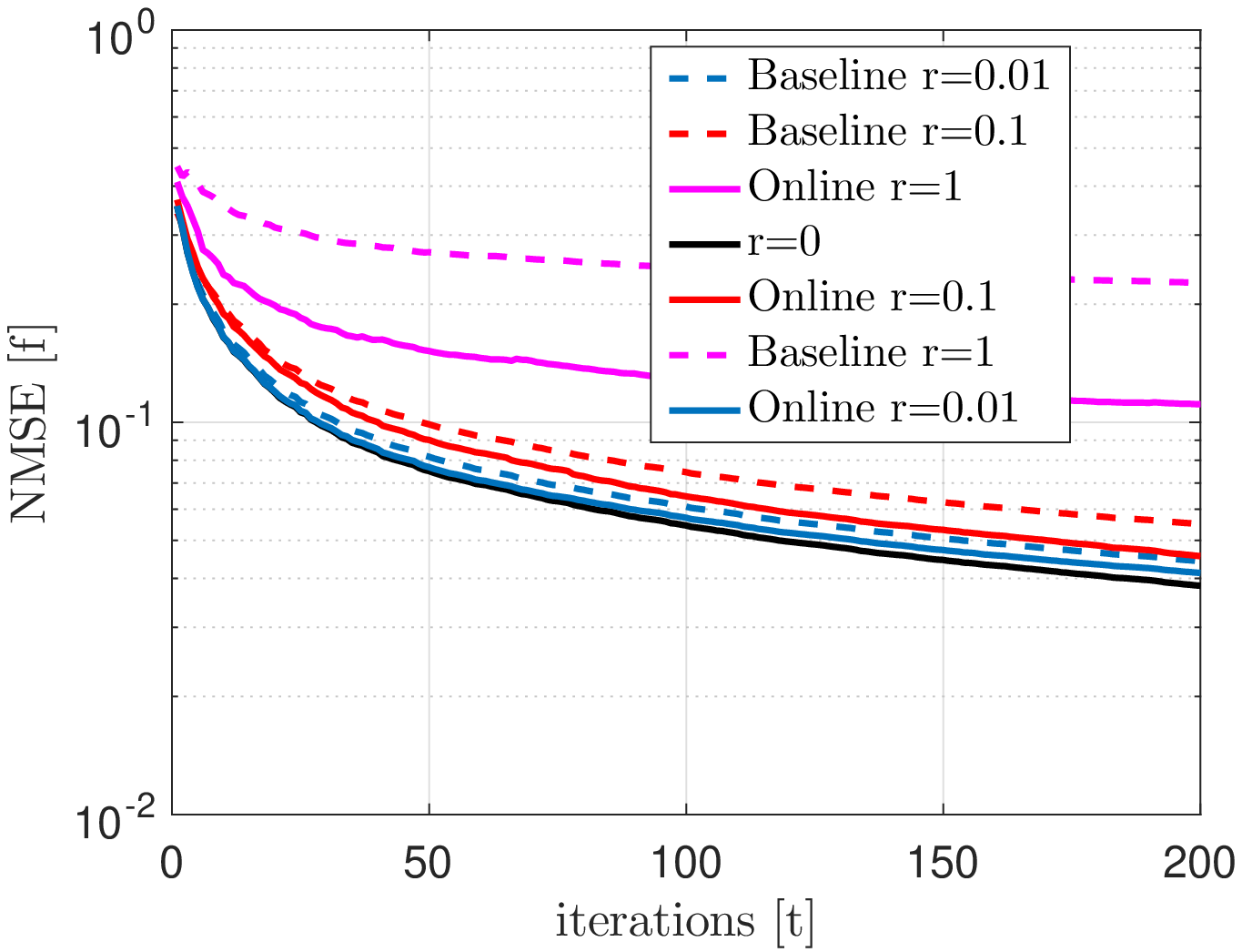}
		\caption{NMSE of $\check{\f}$ vs. $t$.}
		\label{fig:f_vs_it}
	\end{subfigure}
	\hfill
	\begin{subfigure}[b]{.24\textwidth}
		\centering
		\includegraphics[width=1\textwidth]{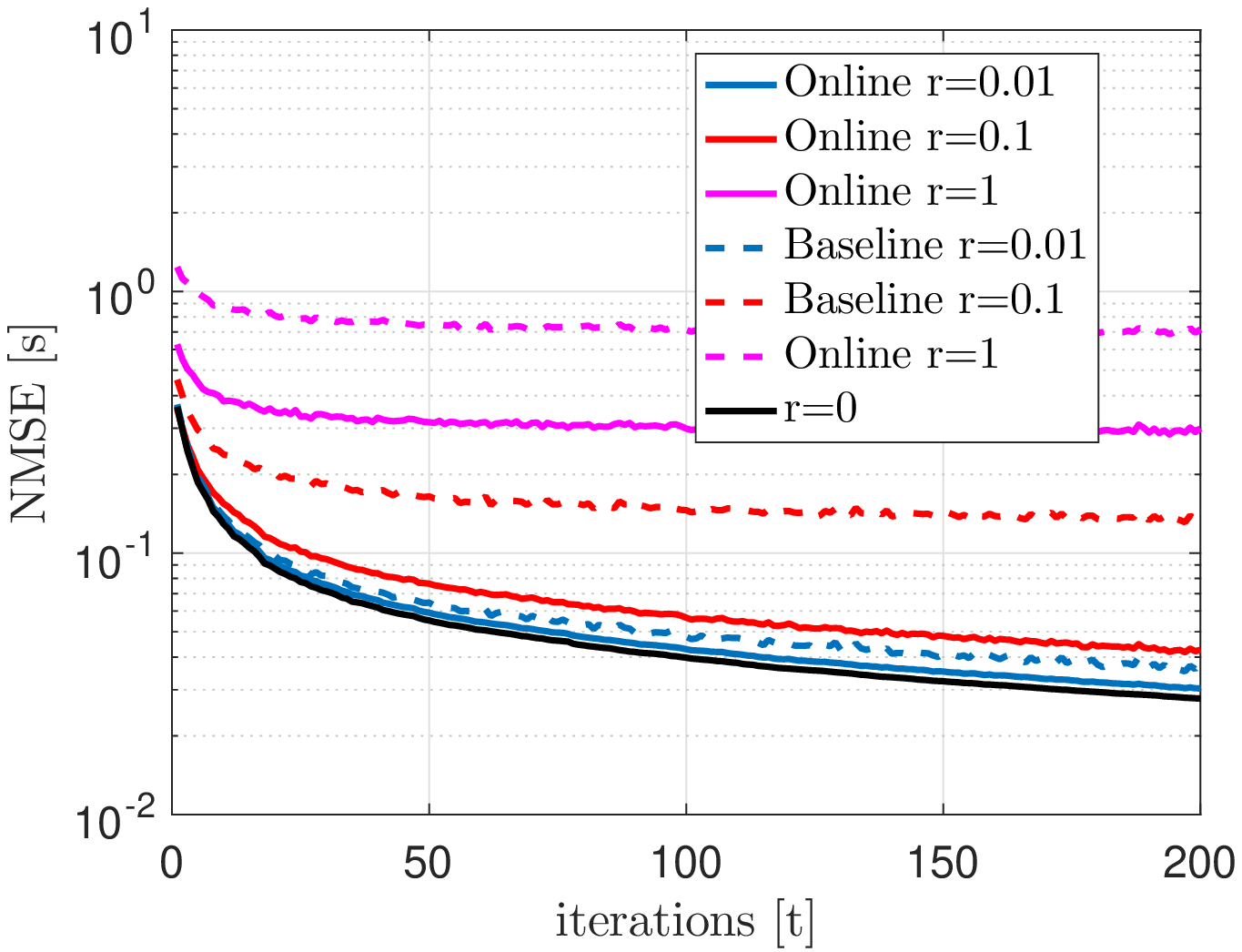}
		\caption{NMSE of $\shahat$ vs. $t$.}
		\label{fig:s_vs_it}
	\end{subfigure}
	\hfill
	\begin{subfigure}[b]{.24\textwidth}
		\centering
		\includegraphics[width=1\textwidth]{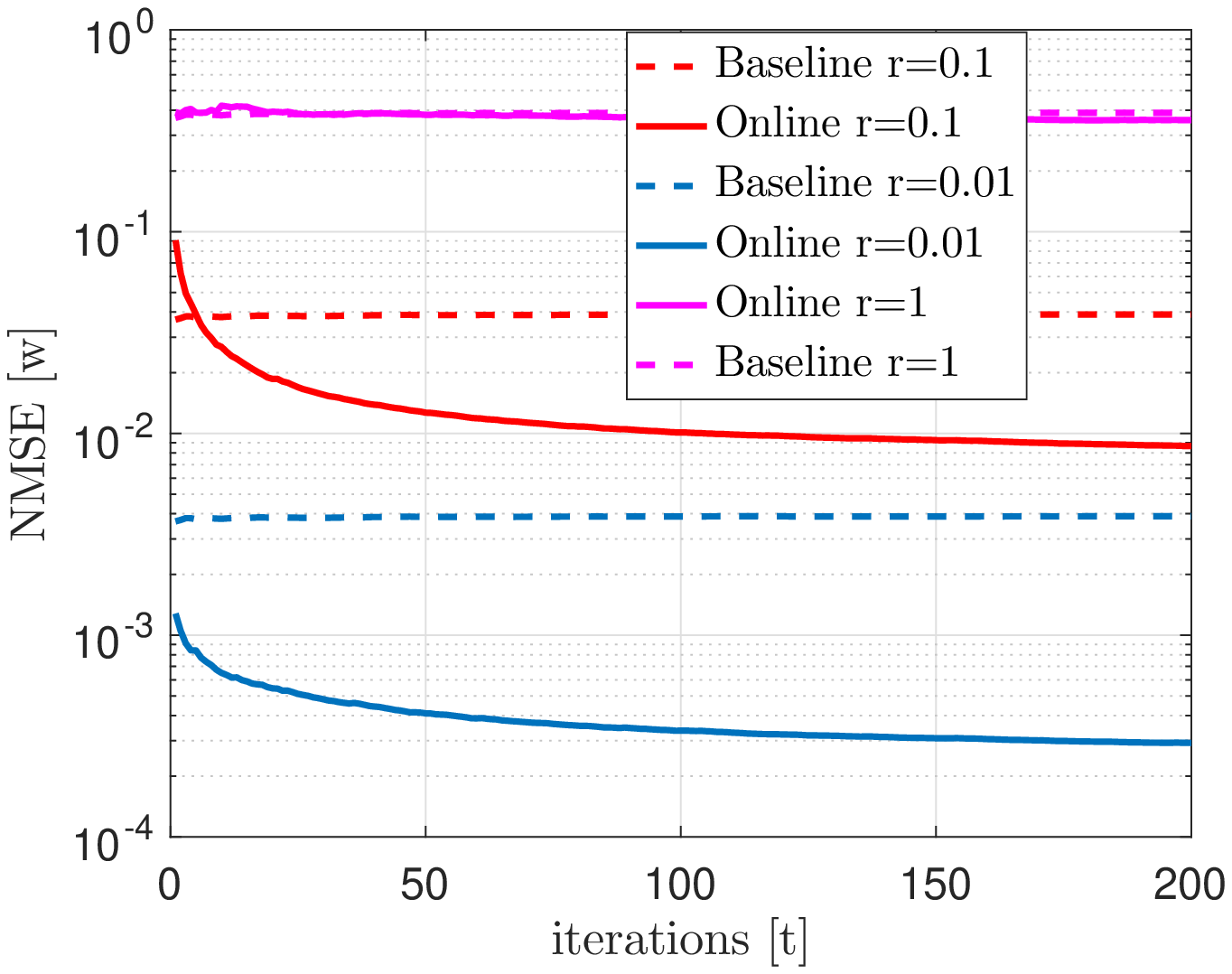}
		\caption{NMSE of $\w$ vs. $t$.}
		\label{fig:w_vs_it}
	\end{subfigure}
	\label{fig:alg_performance}
	\caption{Performance evaluation of the online algorithm and the baseline one for different values of $r$ over iterations $t$.}
	\squeezeup
	\squeezeup
\end{figure}
\begin{figure*}
	\centering
	
	\begin{subfigure}[b]{.3\textwidth}
		\centering
		\includegraphics[width=1\textwidth]{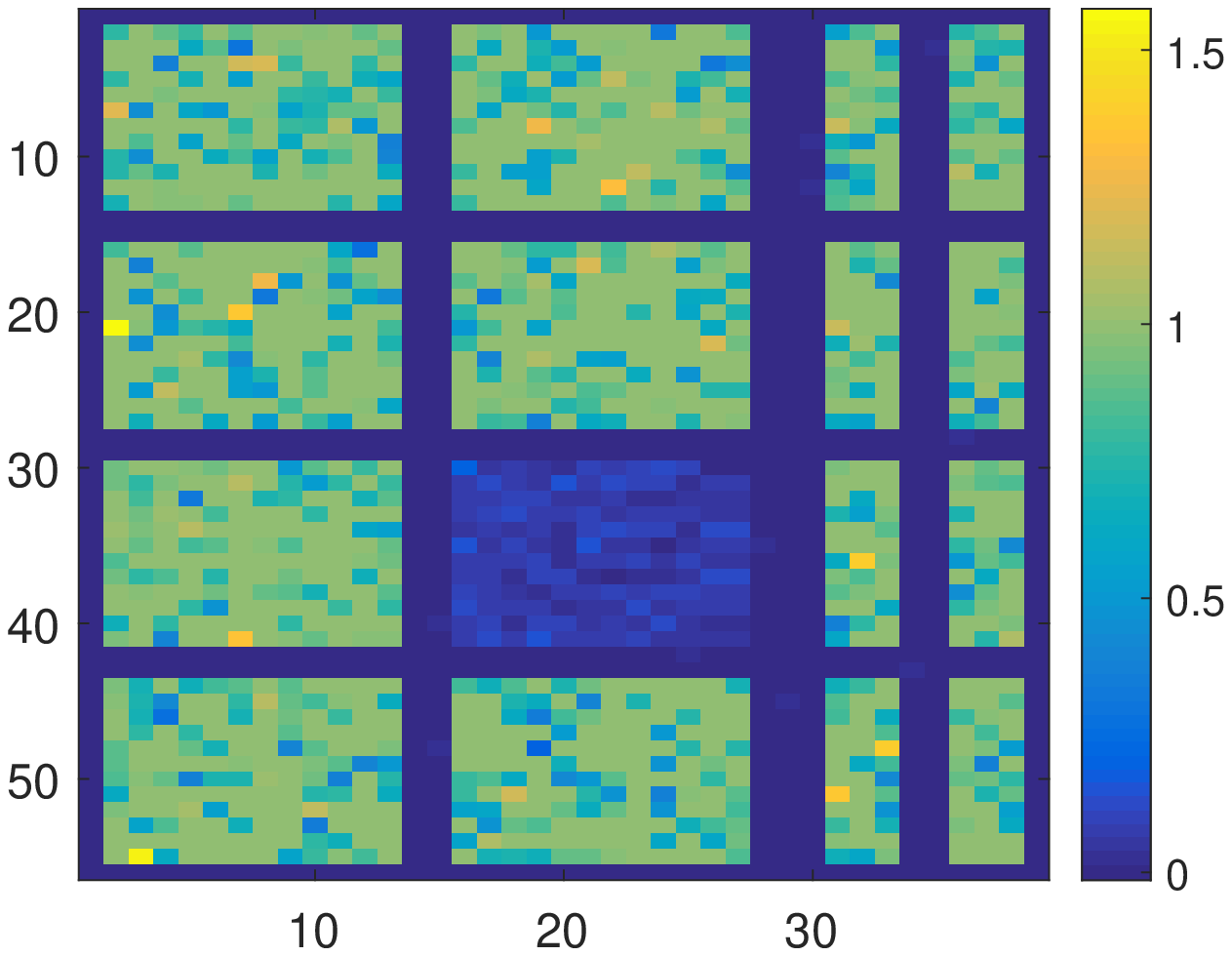}
		\caption{Reconstructed \ac{SLF} for the online algorithm with $r=0.01$.}
		\label{fig:F_online_r0.01}
	\end{subfigure}
	\hfill
	\begin{subfigure}[b]{.3\textwidth}
		\centering
		\includegraphics[width=1\textwidth]{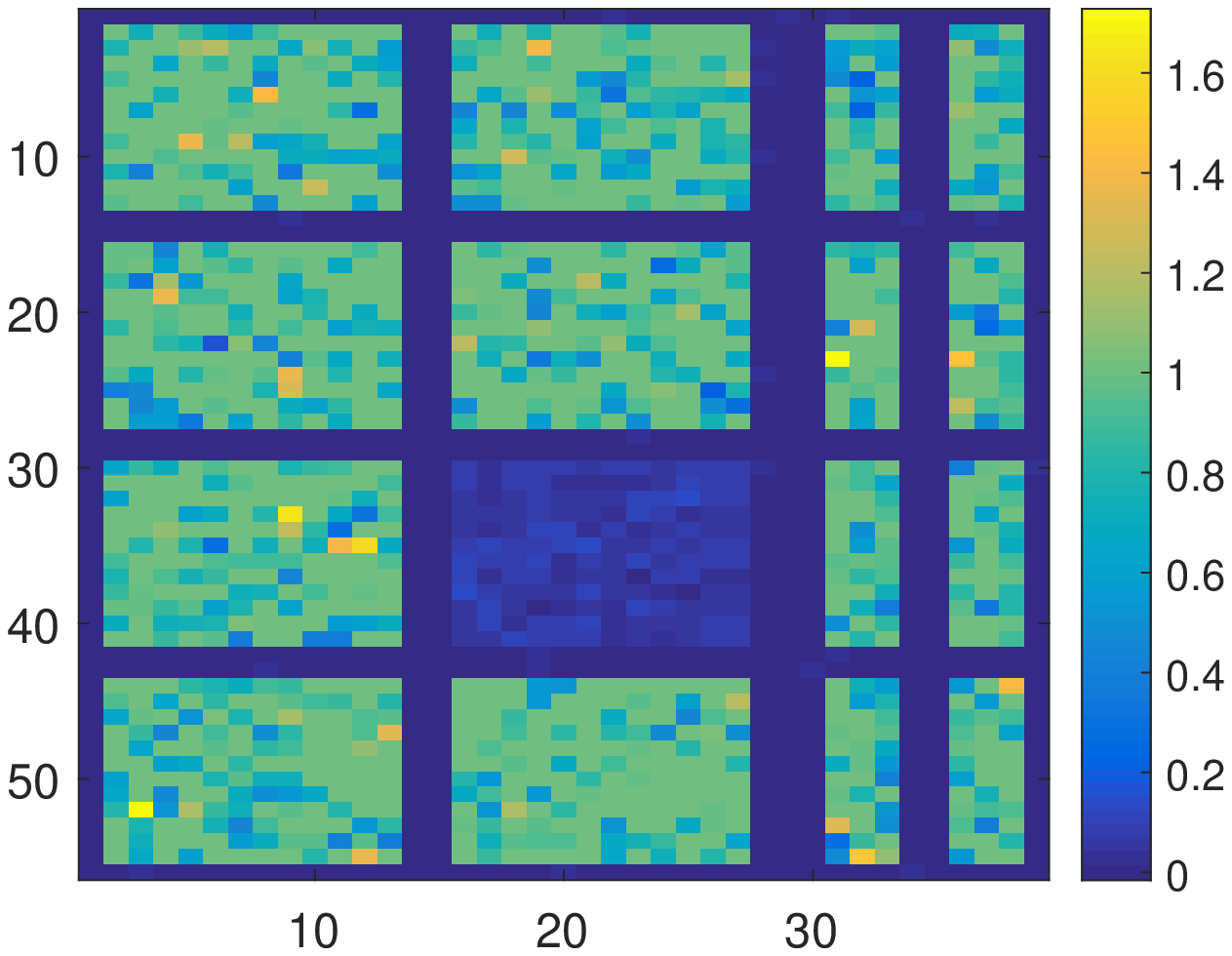}
		\caption{Reconstructed \ac{SLF} for the online algorithm with $r=0.1$.}
		\label{fig:F_online_r0.1}
	\end{subfigure}
	\hfill
	\begin{subfigure}[b]{.3\textwidth}
		\centering
		\includegraphics[width=1\textwidth]{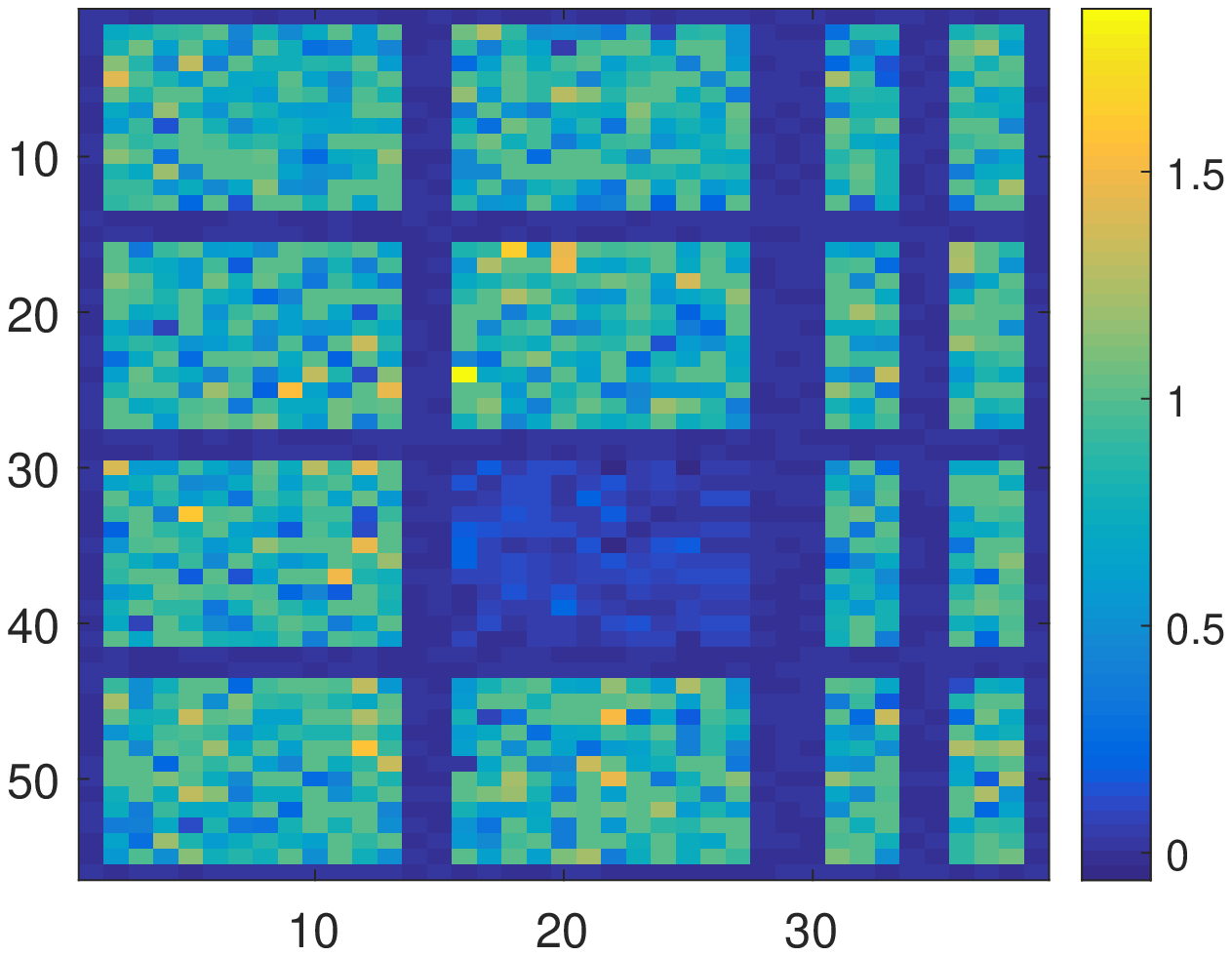}
		\caption{Reconstructed \ac{SLF} for the online algorithm with $r=1$.}
		\label{fig:F_online_r1}
	\end{subfigure}
	\hfill
	\begin{subfigure}[b]{.3\textwidth}
		\centering
		\includegraphics[width=1\textwidth]{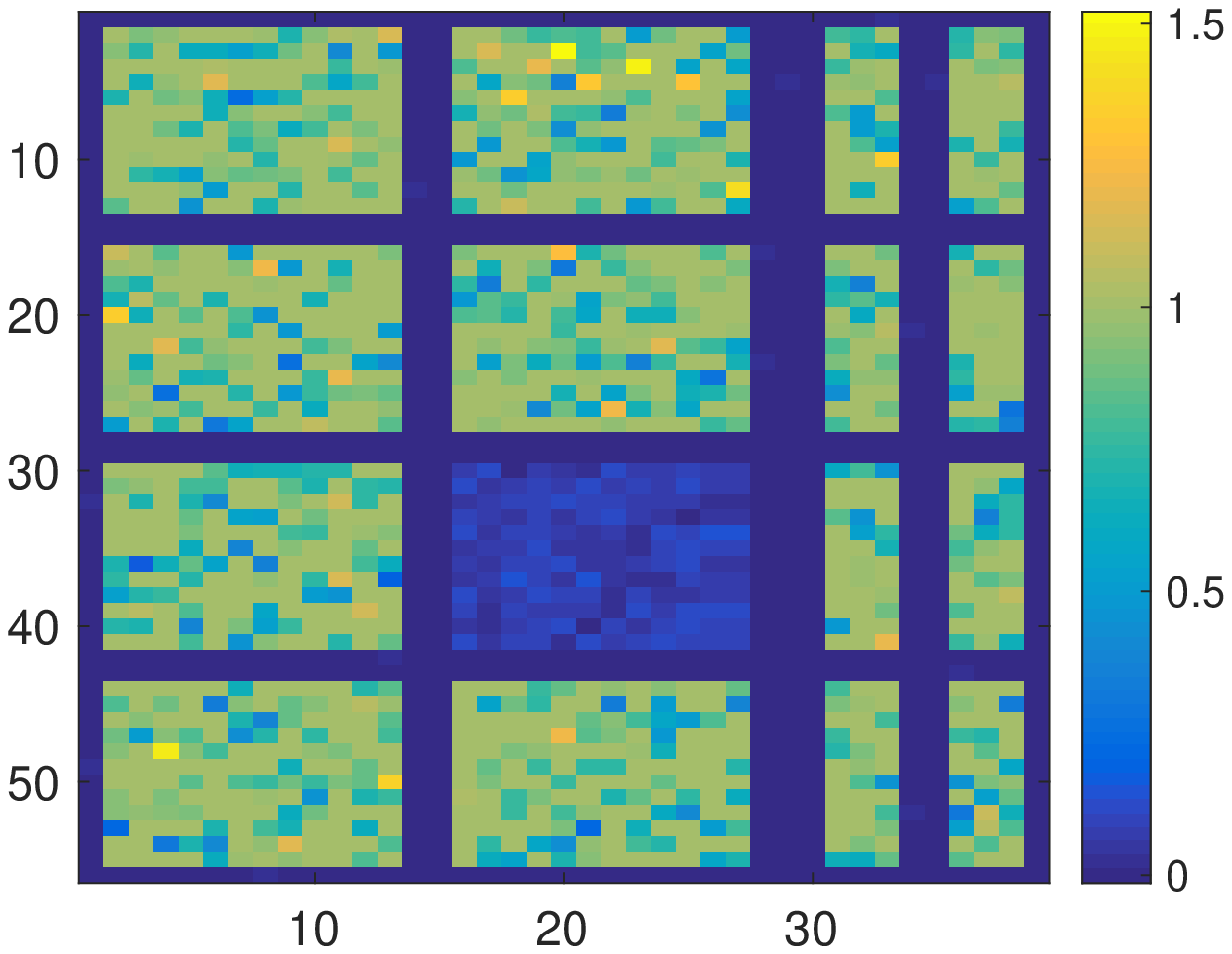}
		\caption{Reconstructed \ac{SLF} for the baseline algorithm with $r=0.01$.}
		\label{fig:F_lee_r0.01}
	\end{subfigure}
	\hfill
	\begin{subfigure}[b]{.3\textwidth}
		\centering
		\includegraphics[width=1\textwidth]{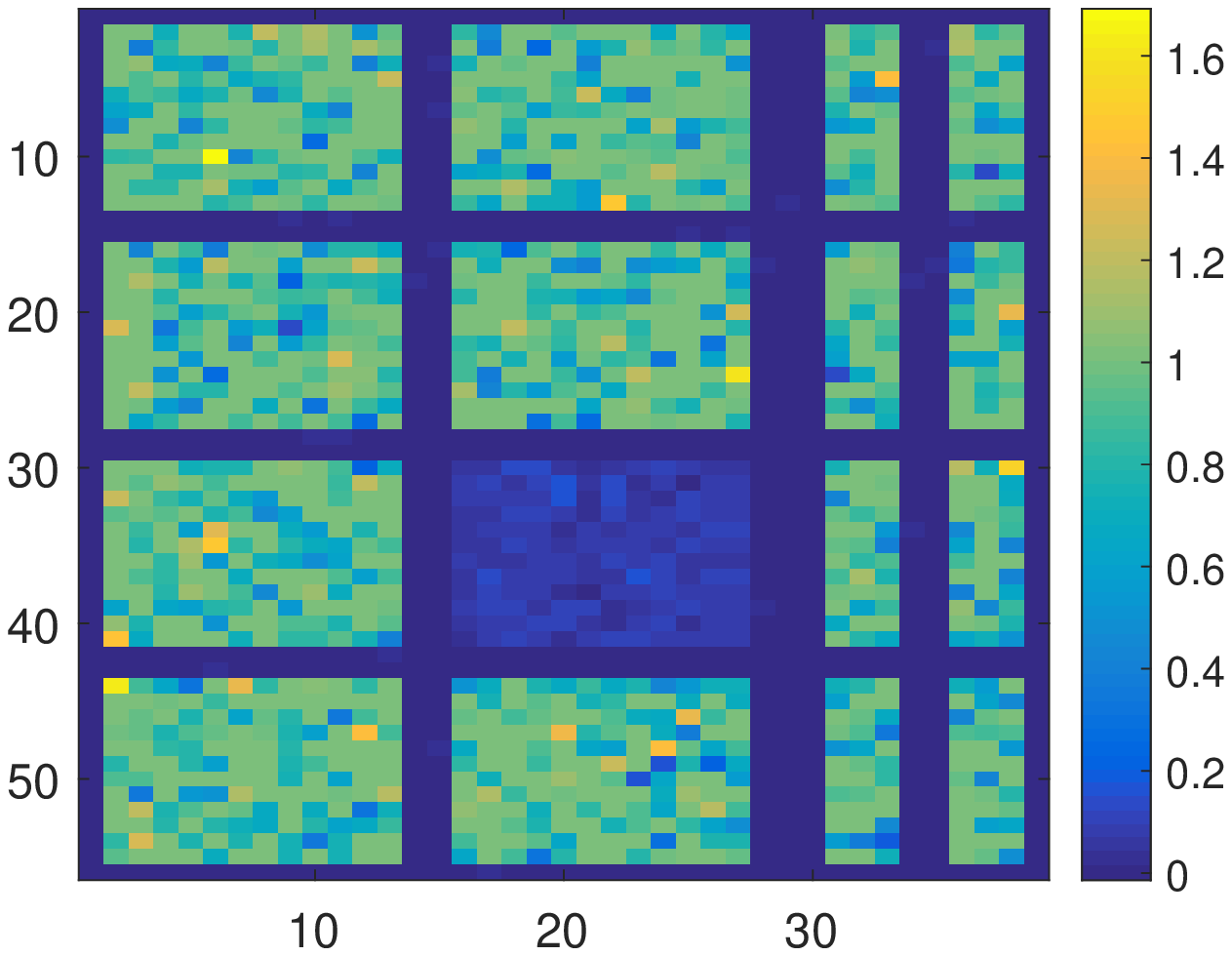}
		\caption{Reconstructed \ac{SLF} for the baseline algorithm with $r=0.1$.}
		\label{fig:F_lee_r0.1}
	\end{subfigure}
	\hfill
	\begin{subfigure}[b]{.3\textwidth}
		\centering
		\includegraphics[width=1\textwidth]{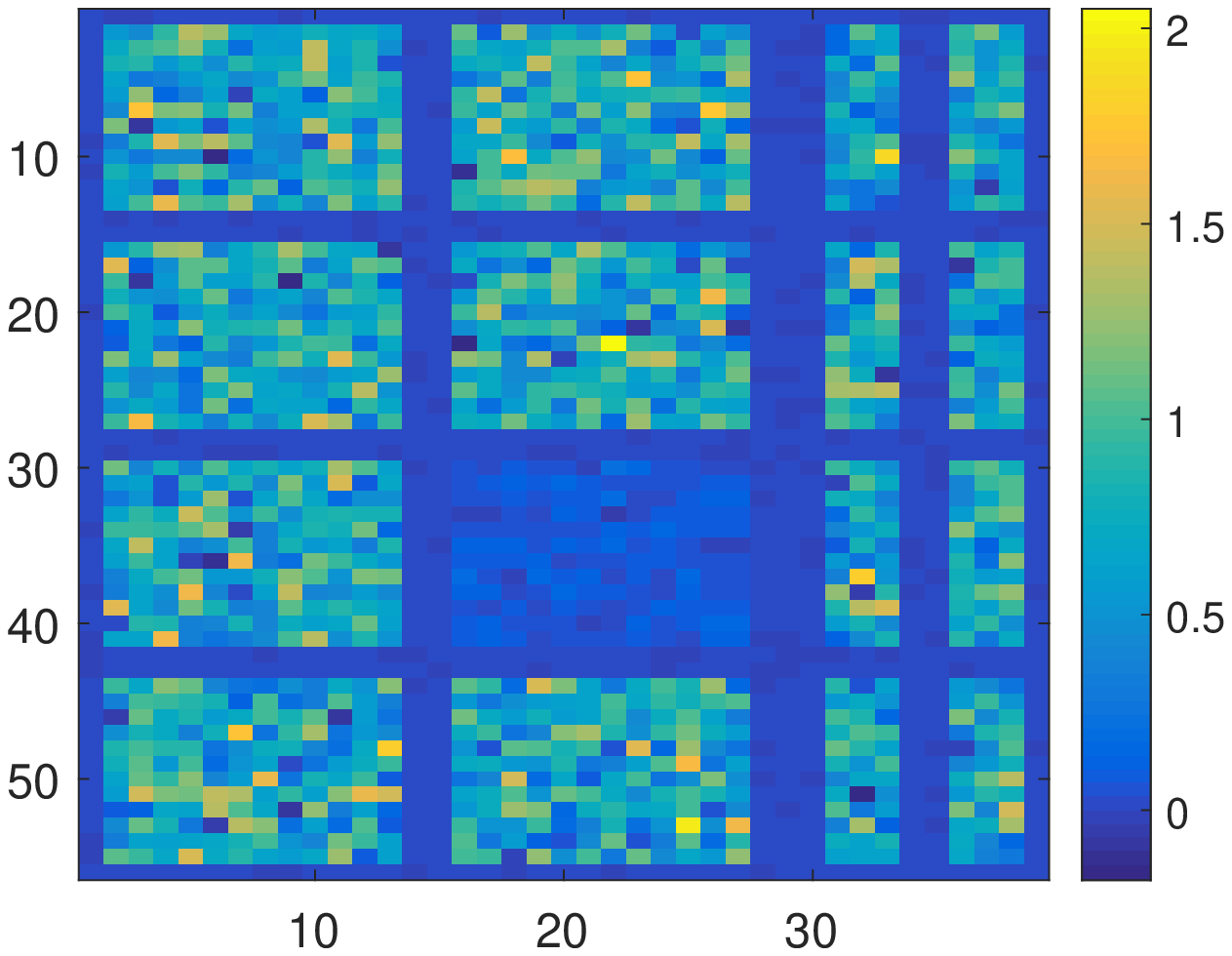}
		\caption{Reconstructed \ac{SLF} for the baseline algorithm with $r=1$.}
		\label{fig:F_lee_r1}
	\end{subfigure}
	\label{fig:F}
	\caption{Reconstructed \ac{SLF} maps for both the online and the baseline algorithm for different values of $r$ after $t=200$ iterations.}
	\squeezeup
	\squeezeup
\end{figure*}

\begin{table}[ht]
	\vskip1mm
	\caption{Simulation parameters with synthetic data}
	\label{table:simulation_parameters}
	\centering
	\begin{tabular}{l l l}
		\hline
		Parameter & Value & Description \\ \hline
		\hline
		$P_x$ & 39 & number horizontal pixels\\
		$P_y$ & 56 &number of vertical pixels\\
		$\eta$ & 0.1499 & wavelength in m\\
		$P$ & 2184 & number of pixels\\
		$T$ & 276396 & number of links in the map\\
		$P_{\textup{tx}}$ & 744 & number of road pixels\\
		$T_{\textup{tx}}$ & 276396 & number of acquirable links\\
		$t_{\textup{max}}$ & 200 & max number of time steps\\
		$M$ & 120 & number of samples acquired per time step\\
		$\sigma$ & 0.0001 & kernel width\\
		$\lambda_1$ & 0.0004 &$\ell_1$ regularization parameter of $\f$\\
		$\lambda_2$ & 0.00001 &$\ell_2$ regularization parameter of $\f$\\
		$\lambda_3$ & 0.00022 &$\ell_2$ regularization parameter of $\alphas_\tau$\\
		\hline
	\end{tabular}
\end{table}

As evaluation metric, we use the \ac{NMSE} of the reconstructed vectors, given by $$\textup{NMSE}(\hat{\signal{v}}):=\frac{\|\hat{\signal{v}}-\signal{v}\|^2_2}{\|\signal{v}\|^2_2},$$
where $\signal{v}$ is any vector and $\hat{\signal{v}}$ its reconstructed version. 

As baseline for the performance comparison, we fix the window matrix $W$ following the elliptical model in \eqref{eq:elliptical_model} and run Alg.~\ref{alg:online} without the projected gradient descent step for the $\alphas$-iterates. We do this to observe the impact of imperfect knowledge of the model, or, in other words, we allow for the $\w_i,i=1,...,P$ to be within a certain radius $r$ of the elliptical model. This version of the algorithm coincides with a slightly modified version of the online algorithm presented in \cite{lee2017channel}, where the authors define the \ac{SLF} structure as the sum of a sparse matrix and a low rank one, and they pose the problem of minimizing the least squares regularized by the sum of the nuclear norm of the low rank matrix and the (matrix) $\ell_1$-norm of the sparse one. As previously mentioned, the main difference between the baseline approach in \cite{lee2017channel} and ours is that they use a fixed model for the window function and assume that the model represents perfectly the reality, while we allow for some flexibility of the model to be within the $\ell_2$-ball of the said structure.

In addition to the aforementioned comparison, we run both the baseline and Alg.~\ref{alg:online} for four different values of the radius $r$, namely, $r=0$, $r=0.01$, $r=0.1$, and $r=1$. We do this to observe the impact on the mismatch between mathematical models and reality, and to examine if our online algorithm can handle this better than other approaches. Note that both algorithms converge to the same solution when the radius is zero given all other parameters are the same, since the sets $\setC_t^\alpha$ are singletons in this case, and the points in these sets coincide with the elliptical model from \eqref{eq:elliptical_model}. 

In Fig.~2, we show the performance evaluation of the algorithms. In particular, Fig.~\ref{fig:o_vs_it} shows the convergence in the objective of Alg.~\ref{alg:online} for different values of $r$. As expected, the cost decreases with the number of iterations in every case, although such decrease is not monotone due to the stochastic nature of the algorithm. Figure \ref{fig:f_vs_it} shows the evolution of the NMSE of the estimated \ac{SLF} vector $\check{\f}$ over the iterations $t$. Of interest is the difference between the baseline algorithm and our approach. We can see that, for $r=0.01,$ both algorithms yield a NMSE close to $r=0,$ but this gap increases more rapidly for the baseline algorithm than for the online one. The difference in performance between both algorithms can be more clearly observed when $r=1$, for which the NMSE of $\check{\f}$ after 200 iterations is around double as much for the baseline algorithm than for our approach. This behavior translates also to the performance in the reconstruction of the shadowing in Fig.~\ref{fig:s_vs_it}, where, again, we see the rapid degradation in accuracy of the baseline when $r$ increases compared to the online algorithm, and, by extension, the path loss is also more accurately reconstructed with the online algorithm. These results validate our intuition that giving flexibility to the original mathematical model can improve the reconstruction performance.

Figure \ref{fig:w_vs_it} shows the NMSE of the reconstructed $\w=[\w_1,...,\w_{t_{\textup{max}}}]^\top$ after converting the estimates $\alphas_1,...,\alphas_{t_{\textup{max}}}$ back into the original space. This is done by multiplying $K_t$ and $\alphas_t$, i.e. $\w_t=K_t\alphas_t.$ We can observe that the reconstructed $\w$ with the baseline algorithm remains flat over $t$ for any value of $r$. This is because in the baseline algorithm, $\w_t$ is fixed and assumed known from the mathematical model. Instead, the NMSE of $\w$ for the online algorithm decreases with the iterations $t$ due to the projected gradient descent strategy to update $\alphas_t$ in \eqref{eq:projected_gradient_its}. Finally, Fig. 3 shows visually the reconstructed \ac{SLF} after 200 iterations for both the baseline and online algorithms with $r=0.01,$ $r=0.1,$ and $r=1.$ Apparently, the reconstructed \acp{SLF} capture the features of the ground-truth \ac{SLF} in Fig.~\ref{fig:slf_map}. However, we can observe the degradation of the \ac{SLF} for the highest value of $r,$ for which the better performance of the online algorithm can be visually stated, i.e., the \ac{SLF} in Fig.~\ref{fig:F_online_r1} looks closer to the ground-truth than that in Fig.~\ref{fig:F_lee_r1}.

Note that Problem \eqref{eq:canonical_f_problem} is underdetermined for all $t$ when $M=120$ with a map of $P=2184$ pixels, since the matrix $A_{\alpha_t}$ has 120 rows and 2184 columns, and therefore it is very flat. This verifies that \ac{A2A} path loss maps can be accurately reconstructed with a small number of measurements by leveraging the group-sparsity of the \ac{SLF}.

\subsection{Evaluation with Realistic V2V Data}


\begin{figure}
	\begin{center}
		\includegraphics[width=.5\textwidth]{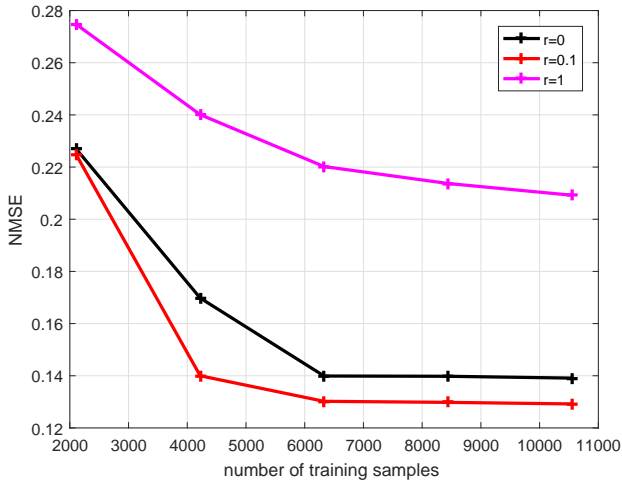}
		\caption{NMSE of the reconstructed shadowing vector $\shahat$ vs. training set sizes.}
		\label{fig:s_vs_it_man}
	\end{center}
\end{figure}

\begin{figure*}
	\centering
	
	\begin{subfigure}[b]{.3\textwidth}
		\centering
		\includegraphics[width=1\textwidth]{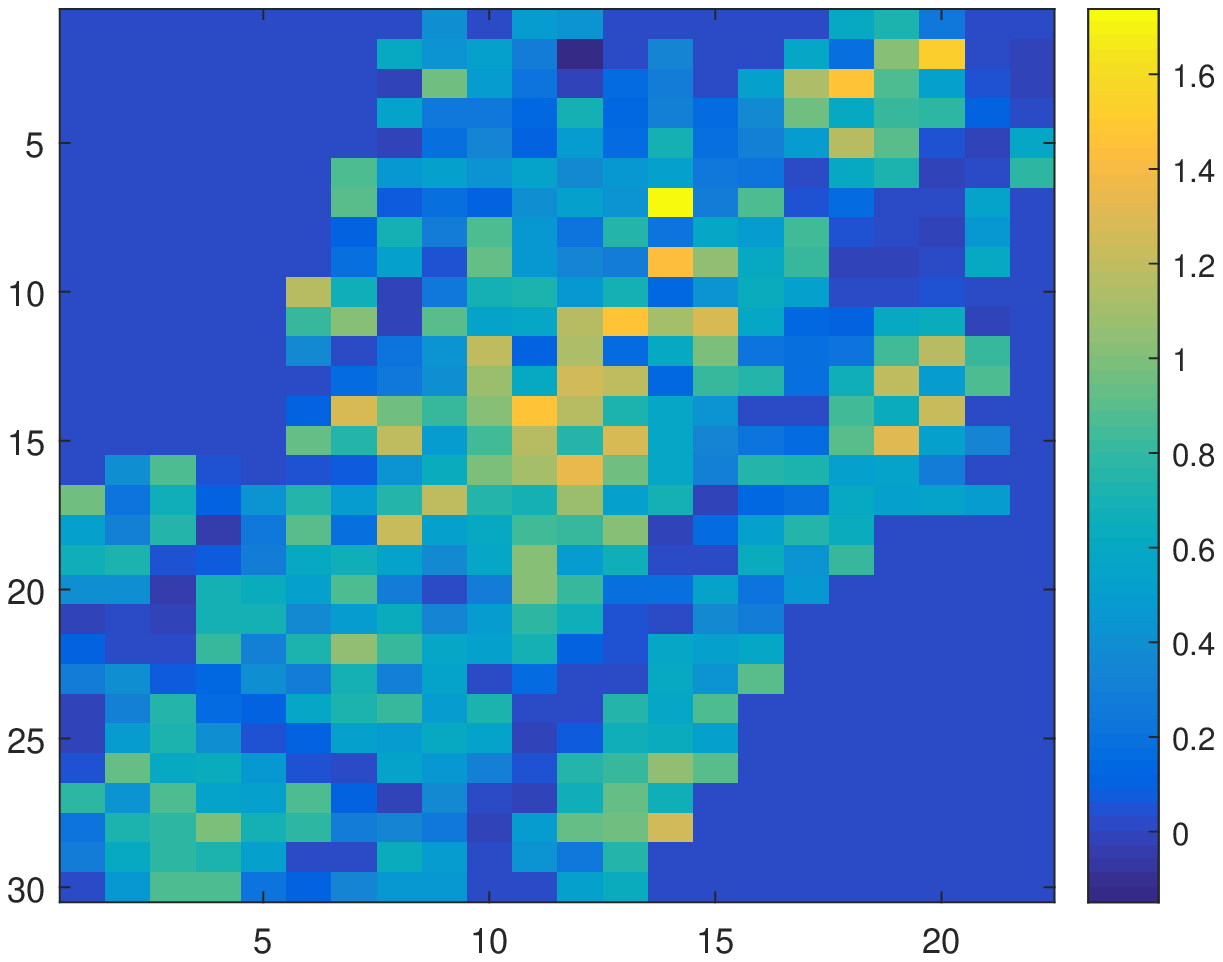}
		\caption{Reconstructed \ac{SLF} with $r=0.1$ and 2110 training samples.}
		\label{fig:F_man_r001.01}
	\end{subfigure}
	\hfill
	\begin{subfigure}[b]{.3\textwidth}
		\centering
		\includegraphics[width=1\textwidth]{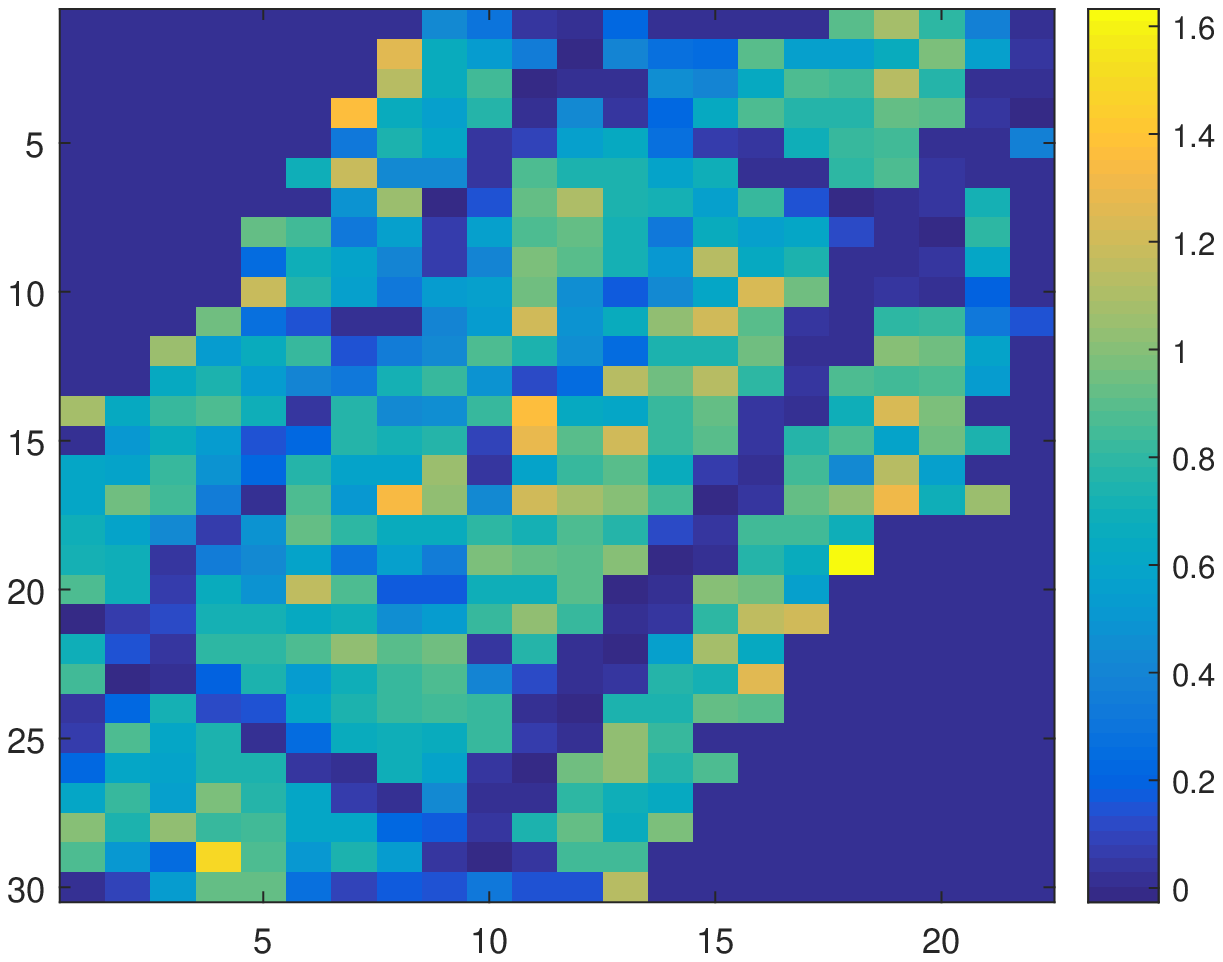}
		\caption{Reconstructed \ac{SLF} with $r=0.1$ and 6330 training samples.}
		\label{fig:F_man_r0.1}
	\end{subfigure}
	\hfill
	\begin{subfigure}[b]{.3\textwidth}
		\centering
		\includegraphics[width=1\textwidth]{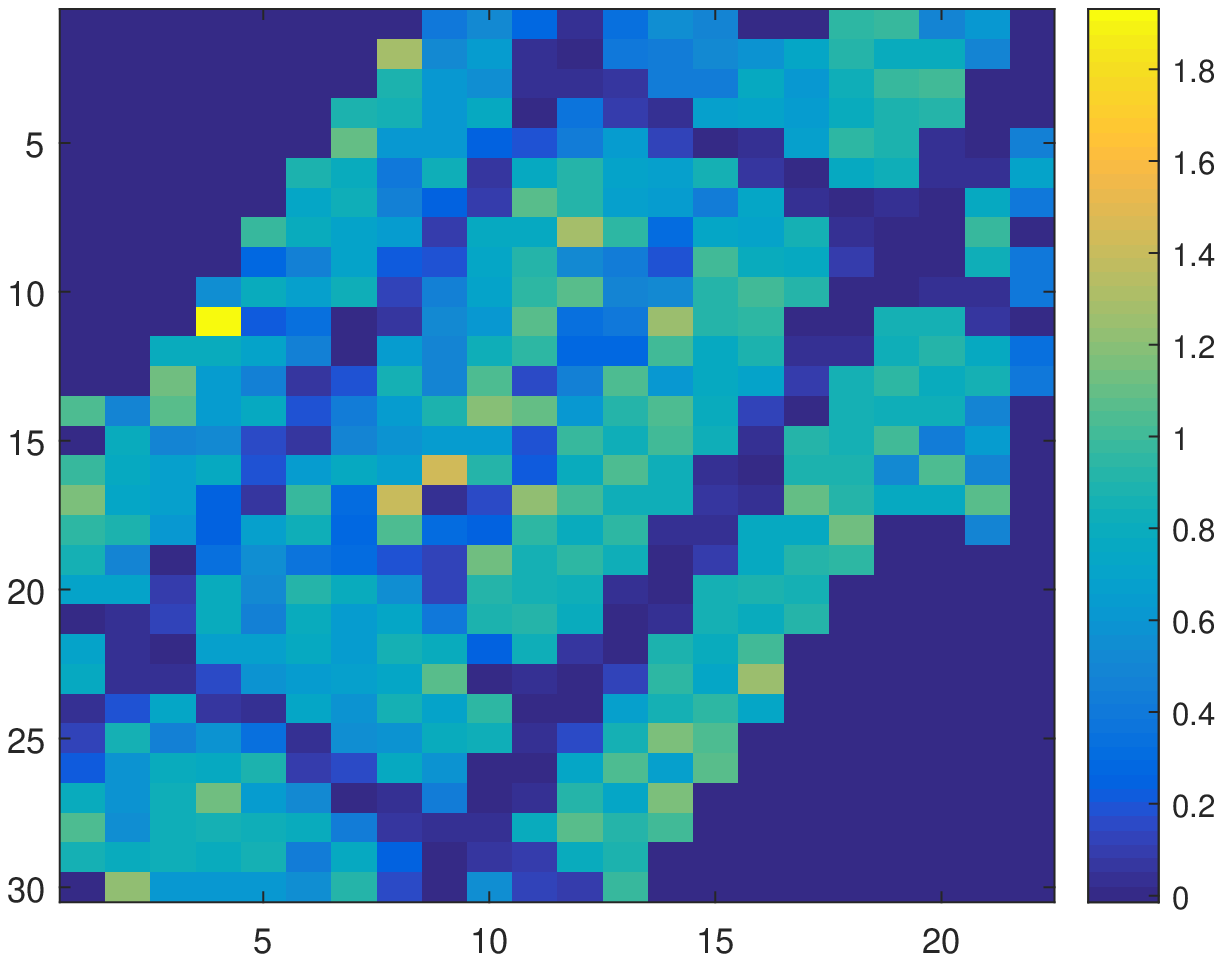}
		\caption{Reconstructed \ac{SLF} with $r=0.1$ and 10551 training samples.}
		\label{fig:F_man_r1}
	\end{subfigure}
	\hfill
	\label{fig:F_man}
	\caption{Reconstructed \ac{SLF} maps for $r=0.1$ with different training set size.}
	\squeezeup
	\squeezeup
\end{figure*}

\ac{GEMV2} adopts location-specific propagation modeling with respect to large objects in the vicinity of the communicating vehicles such as buildings and foliage. More specifically, the model uses the real-world locations and dimensions of nearby buildings, foliage, and vehicles to determine the \ac{LOS} or \ac{NLOS} conditions for each link. Starting from this premise, \ac{GEMV2} uses simple geographical descriptors of the simulated environment (outlines of buildings, foliage, and vehicles on the road) to classify V2V links into three groups, namely, \ac{LOS} links, NLOSv, i.e., links whose LOS is obstructed by other vehicles; and NLOSb, i.e., links whose LOS is obstructed by building or foliage. Based on this link classification, GEMV2 deterministically calculates the large-scale signal variation (i.e., path loss and shadowing) for each link type, and it adds a fast-fading term based on a particular random distribution. As shown in \cite{boban2014geometry}, the model fits real-life \ac{V2V} measurements for different urban scenarios well, making the model a good proxy for the evaluation of our algorithm with a realistic path loss dataset.

We collected a dataset of \ac{V2V} received power measurements based on the GEMV2 model for a map located in the Lower-Manhattan area. To generate the dataset, we simulated the communication of only two vehicles moving around the map. The reason behind this is the aim to focus the study on the influence on the path loss coming from free space attenuation and fixed objects, so the scenario remains as static as possible. 


	

Before feeding the training data to the algorithm, the shadowing of each link has to be derived from the dataset of received powers generated by the GEMV2 software. To this end, we first calculate the path loss from the received power values, given by
\begin{equation}
	\label{eq:pl_from_power}
	\mathrm{pl}^{(t)}=p_{\mathrm{tx}}-p^{(l)}_{\mathrm{rx}},
\end{equation}
where $p_{\mathrm{tx}}$ is the transmit power, and $p^{(l)}_{\mathrm{rx}}$ is the received power in decibels of link index $l$. The shadowing is then obtained from Eq.~\eqref{eq:pathloss} by subtracting the free space path loss given by 
$$ \mathrm{pl}^{(t)}_{\mathrm{free}}=\mathrm{pl}_0+10\delta\log_{10}\left(\frac{||\mathbf{x}^{(t)}_i-\mathbf{x}^{(t)}_j||_2}{d_0}\right) $$
to the path loss obtained from Eq.~\eqref{eq:pl_from_power}. In our case, $\delta=2.9$ and $\mathrm{pl}_0=75$ dB. The dataset contains 21103 unique measurements representing received powers of links connecting only road locations. The frequency band is $5.89$ GHz, transmit power is $12$ dBm, and we assign each link to a Tx/Rx pair in a grid of $30\times 22$ pixels, with each pixel having $6\times 6$ squared meters.

We train our algorithm with five different training sets randomly selected from the shadowing dataset. The training sets contain 2110, 4220, 6330, 8441 and 10551 samples, which correspond with $10\%,~20\%,~30\%,~40\%,$ and $50\%$ of the available samples, respectively. After performing 3-fold cross-validation in the largest training set, the algorithm parameters are set to $\sigma= 0.0001$, $\lambda_1=0.0006$, $\lambda_2=0.00001$, and $\lambda_3=0.00061$. We set $t_{\mathrm{max}}=200$ and, at each time instant, $M=60$ i.i.d. samples are drawn from the training set. Table \ref{table:simulation_parameters_2} summarizes the main simulation parameters.
\begin{table}
	\vskip1mm
	\caption{Simulation parameters with realistic V2V data}
	\label{table:simulation_parameters_2}
	\centering
	\begin{tabular}{l l l}
		\hline
		Parameter & Value & Description \\ \hline
		\hline
		$P_x$ & 30 & number horizontal pixels\\
		$P_y$ & 22 &number of vertical pixels\\
		$B$ & 5.89 & frequency band [GHz]\\
		$P$ & 660 & number of pixels\\
		$T$ & 217470 & number of links in the map\\
		$P_{\textup{tx}}$ & 129 & number of road pixels\\
		$D$ & 21103 & size of dataset\\
		$t_{\textup{max}}$ & 200 & max number of time steps\\
		$M$ & 60 & number of samples acquired per time step\\
		$\delta$ & 2.9 & path loss exponential decay\\
		$\mathrm{pl}_0$ & 75 & path loss at reference distance [dB]\\
		$p_{\mathrm{tx}}$ & 12 & transmit power [dBm]\\
		$\sigma$ & 0.0001 & kernel width\\
		$\lambda_1$ &0.0006 &$\ell_1$ regularization parameter of $\f$\\
		$\lambda_2$ & 0.00001 &$\ell_2$ regularization parameter of $\f$\\
		$\lambda_3$ & 0.00061 &$\ell_2$ regularization parameter of $\alphas_\tau$\\
		\hline
	\end{tabular}
\end{table}
In order to assess the usefulness of the hybrid model and data driven approach, the experiment is carried out three times for each of the training sets, each time with a radius $r$ of $0$, $0.1$ and $1$. Note that, as in the evaluation in Section \ref{sec:num_ev_synthetic}, $r=0$ corresponds to our baseline evaluation for which the model is assumed to represent perfectly the physical reality, whereas the experiments with radius $r>0$ allow for misalignments of the model and reality. 

Figure \ref{fig:s_vs_it_man} shows the NMSE of the reconstructed shadowing vector $\shahat$ versus the size of the training dataset, for $r=0$, $r=0.1$ and $r=1$. As expected, the larger the training set, the better the performance. However, the improvement between 6330 and 10551 training samples is small, which hints that a reasonable reconstruction performance can be achieved even for small training set sizes. The experiment with $r=0.1$ clearly performs the best for all numbers of training samples, and the gap w.r.t. the model-based baseline $(r=0)$ seems to increase with larger training datasets. On the other hand, the experiments with $r=1$ perform the worst by a big margin. These results convey two take-aways: i) that the hybrid approach presented in this paper can improve the performance of a model-based approach in real scenarios (i.e., $r=0$), and ii) that the selection of the radius $r$ is critical, since the performance can significantly worsen compared to a model-based approach, for large values of $r$.

Figures 7a-7c show the reconstructed normalized \ac{SLF} for the experiments with $r=0.1$ and different training set sizes. We can clearly see the resemblance of the area to the considered map even for the smallest training set size, while there is no big differences between the experiments with $30\%$ (Fig. 7b) and $50\%$ (Fig. 7c) of the samples for training, which suggests that the proposed algorithm is capable of reconstructing the \ac{SLF} in realistic scenarios. 

\section{Conclusions}
In this paper, we have addressed the online learning of path loss maps through a hybrid model and data driven approach. In order to estimate the shadowing experienced by a radio link connecting any two locations in a map, we have formulated a problem to simultaneously obtain an estimate of both the \ac{SLF} and the model from \ac{TPT}. We have considered the elastic net as regularization because the \ac{SLF} is assumed to be group-sparse. The resulting problem is highly ill-posed, so we have added structure by considering a non-linear kernel approach. 

We have proposed an online algorithm based on stochastic optimization and alternating minimization to tackle the high complexity of the problem even for small maps, and we have proven the convergence of the online algorithm both in the objective and in the arguments. Finally, we have shown by simulations with synthetic data as well as with realistic data that the proposed method outperforms other state-of-the-art techniques when the data do not fit perfectly the model.



\begin{acronym}[SPACEEEEEE]
	\setlength{\itemsep}{-0.2em}
	\acro{1PPS}{one pulse per second}
	\acro{3GPP}{Third Generation Partnership Project}
	\acro{3G}{third generation}
	\acro{4G}{fourth generation}
	\acro{ACK/NACK}{(not) acknowledgements}
	\acro{aGW}{advanced gateway}
	\acro{AoA}{angle of arrival}
	\acro{AoD}{angle of departure}
	\acro{AMC}{adaptive modulation and coding}
	\acro{ARQ}{automatic repeat request}
	\acro{ASIC}{application-specific integrated circuit}
	\acro{AGC}{automatic gain control}
	\acro{AWGN}{additive white Gaussian noise}
	\acro{BC}{broadcast channel}
	\acro{BER}{bit error rate}
	\acro{BICM}{bit-interleaved coded modulation}
	\acro{BPSK}{binary phase-shift keying}
	\acro{BS}{base station}
	\acro{BOF}{beginning of frame}
	\acro{BUC}{block up-converter}
	\acro{CAPEX}{capital expenditure}
	\acro{CDMA}{code-division multiple access}
	\acro{CC}{chase combining}
	\acro{CID}{cell identified}
	\acro{CIR}{channel impulse response}
	\acro{CU}{central unit}
	\acro{CUBA}{circular uniform beam array}
	\acro{CSI-RS}{CSI reference signals}
	\acro{CCI+}[CCI]{Cochannel interference}
	\acro{CCI}{cochannel interference}
	\acro{cdf}[CDF]{cumulative distribution function}
	\acro{CFO}{carrier frequency offset}
	\acro{CFR}{channel frequency response}  
	\acro{CLE}{chip-level equalizer}
	\acro{CCDF}{complementary cumulative distribution function}
	\acro{CDM}{code-division multiplexing}
	\acro{CoMP}{coordinated multi-point}
	\acro{CoSCH}{coordinated scheduling}
	\acro{CMF}{code-matched filter}
	\acro{CQI}{channel quality identifier}
	\acro{CP}{cyclic prefix}
	\acro{CO}{central office}
	\acro{CPE}{customer-provided equipment}
	\acro{CRC}{cyclic redundancy check}
	\acro{CRS}{CSI reference signals}
	\acro{CSI}{channel state information}
	\acro{CPE}{common phase error}
	\acro{CPICH}{common pilot channel}
	\acro{CPRI}{common public radio interface}
	\acro{CWER}{code word error rate}
	\acro{DFT}{discrete Fourier transform}
	\acro{DC}{direct current}
	\acro{DD}{digital dividend}
	\acro{DS}{delay spread}
	\acro{DMMT}{discrete matrix multi-tone}
	\acro{EVM}{error vector magnitude}
	\acro{DFT}{discrete Fourier transform}
	\acro{DoD}{direction of departure}
	\acro{DoA}{direction of arrival}
	\acro{DMMT}{discrete matrix multi-tone}
	\acro{DSSS}{direct sequence spread spectrum}
	\acro{DSP}{digital signal processor}
	\acro{DSL}{digital subscriber line}
	\acro{DS-UWB}{direct sequence ultra-wideband}
	\acro{DRS}{demodulation reference signals}
	\acro{ED}{excess delay}
	\acro{EGT}{equal gain transmission}
	\acro{EGC}{equal gain combining}
	\acro{ERP}{effective radiated power}
	\acro{EO}{electro-optical}
	\acro{FDE}{frequency-domain equalization}
	\acro{FA}{frequency advance}
	\acro{FD}{frequency domain}
	\acro{FDD}{frequency division duplex}
	\acro{FIR}{finite impulse response}
	\acro{FWHM}{full width at half maximum}
	\acro{FDMA}{frequency-division multiple access}
	\acro{FCC}{Federal Communications Commission}
	\acro{FEC}{forward error correction}
	\acro{FFT}{fast Fourier transform}
	\acro{FSK}{frequency shift keying}
	\acro{FR}{frequency response}
	\acro{FTTH}{fiber to the home}
	\acro{FOV}{field of view}
	\acro{FPGA}{field programmable gate array}
	\acro{GoB}{grid of beams}
	\acro{GI}{guard interval}
	\acro{GF}{geometry factor}
	\acro{GPS}{global positioning system}
	\acro{GSM}{global system for mobile communications}
	\acro{HARQ}{hybrid automatic repeat request}
	\acro{HHI}{Heinrich Hertz Institute}
	\acro{HFT}{Institut f{\"u}r Hochfrequenztechnik}
	\acro{HSDPA}{High-Speed Downlink Packet Access}
	\acro{HSOPA}{High Speed OFDM Packet Access}
	\acro{HOSVD}{Higher Order Singular Value Decomposition}
	\acro{IFFT}{inverse fast Fourier transform}
	\acro{ICI}{inter-carrier interference}
	\acro{IDFT}{inverse discrete Fourier transform}
	\acro{iid}[i.i.d.]{independent and identically distributed}
	\acro{IF}{intermediate frequency}
	\acro{IIR}{infinite impulse response}
	\acro{IR}{impulse response}
	\acro{MAC}{multiple-access control}
	\acro{IRC++}[IRC]{Interference Rejection Combining}
	\acro{IRC+}[IRC]{Interference rejection combining}
	\acro{IRC}{interference rejection combining}
	\acro{ILR}{Institut f{\"u}r Luft- und Raumfahrt}
	\acro{ISD}{inter-site distance}
	\acro{ISI}{intersymbol interference}
	\acro{IP}{internet protocol}
	\acro{JT}{joint transmission}
	\acro{JT CoMP}{joint transmission coordinated multi-point}
	\acro{LDC}{linear dispersion code}
	\acro{L2S}{link-to-system} 
	\acro{LAN}{local area network}
	\acro{LMMSE}{linear minimum mean square error}
	\acro{LOS}{line-of-sight}
	\acro{LO}{local oscillator}
	\acro{LSU}{LTE signal processing unit}
	\acro{LTE}{Long Term Evolution}
	\acro{LTE-A}{LTE-Advanced}
	\acro{LUT}{look-up table}
	\acro{MATH}{Institut f{\"u}r Mathematik}
	\acro{MAC}{medium access control}
	\acro{MAI}{multiple access interference}
	\acro{MAC+}[MAC layer]{medium access layer}
	\acro{maxSINR}{maximum SINR}
	\acro{MCS}{modulation and coding scheme}
	\acro{MB-OFDM}{multi-band orthogonal frequency division multiplexing}
	\acro{MFN}{multi frequency network}
	\acro{MIESM}{mutual information effective SINR metric}
	\acro{MIMO}{multiple-input multiple-output}
	\acro{MMHARQ}[MM-HARQ]{MIMO multiple HARQ}
	\acro{ML}{maximum likelihood}
	\acro{MRC}{maximum ratio combining}
	\acro{MSHARQ}[MS-HARQ]{MIMO single HARQ}
	\acro{MSE}{mean square error}
	\acro{MMSE}{minimum mean square error}
	\acro{MLSE}{maximum likelihood sequence estimation}  
	\acro{MMSE++}{Minimum Mean Square Error}
	\acro{MPLS}{multi-protocol label switching}
	\acro{MSE}{mean square error}
	\acro{MS}{multiple stream}
	\acro{MT}{mobile terminal}
	\acro{MT-S}[MT scheduler]{maximum throughput scheduler}
	\acro{MU}{multi-user}
	\acro{MU-SDMA}{multi-user space-division multiple access}
	\acro{MU-MUX}{multi-user spatial multiplexing}
	\acro{NGMN}{next generation mobile network}
	\acro{NLOS}{non line-of-sight}
	\acro{NMEA}{National Marine Electronics Association}
	\acro{NT}[$N_T$]{number of transmit antennas}
	\acro{OVSF}{orthogonal variable spreading factor} 
	\acro{OE}{opto-electrical}
	\acro{OFDM}{orthogonal frequency-division multiplexing}
	\acro{OFDMA}{orthogonal frequency division multiple access}
	\acro{OOK}{on-off keying}
	\acro{OC}{optimum combining}
	\acro{OCXO}{oven-controlled crystal oscillator}
	\acro{OPEX}{operational expenditure}
	\acro{PA}{power amplifier}
	\acro{PAM}{pulse amplitude modulation}
	\acro{PARC}{per antenna rate control}
	\acro{PAPC}{per antenna power constraint}
	\acro{PAPR}{peak to average power ratio}
	\acro{PMCC}{Pearson product-moment correlation coefficient}
	\acro{PER}{packet error rate}
	\acro{PL}{path loss}
	\acro{PDP}{power delay profile}
	\acro{PDF}{probability density function}
	\acro{PUCA}{polarized uniform circular array}
	\acro{PF-S}[PF scheduler]{proportional fair scheduler}
	\acro{PMI}{precoding matrix indicator}
	\acro{PHY}{physical layer}
	\acro{PDP}{power delay profile}
	\acro{PDU}{packet data unit}
	\acro{PDCCH}{physical downlink control channel}
	\acro{PUCCH}{physical uplink control channel}
	\acro{PUSCH}{physical uplink shared channel}
	\acro{PPM}{pulse position modulation}
	\acro{PON}{passive optical network}
	\acro{PPS}{pulse per second}
	\acro{PRS}{pseudo-random scrambling sequence}
	\acro{PSS}{primary synchronization sequence}
	\acro{PSD}{power spectral density}
	\acro{QAM}{quadrature amplitude modulation}
	\acro{QPSK}{quadrature phase-shift keying}
	\acro{QoS}{quality of service}
	\acro{RAN}{radio access network}
	\acro{RACH}{random access channel}
	\acro{RD}{rate-distortion}
	\acro{RR}{round robin}
	\acro{RoT}{rise-over-thermal}
	\acro{RF}{radio frequency}
	\acro{RFO}{reference frequency offset}
	\acro{RS}{reference signals}
	\acro{RB}{resource block}
	\acro{Rx}{receive}
	\acro{RMS}{root mean square}
	\acro{RRH}{remote radio head}
	\acro{RRC}{root raised cosine}
	\acro{RTS}{real time sampled}
	\acro{SAE}{system architecture evolution} 
	\acro{SG}{scenario group}
	\acro{SB-S}[SB scheduler]{score-based scheduler}
	\acro{SC}{sub-carrier}
	\acro{SCM}{spatial channel model}
	\acro{sc}{single-carrier}
	\acro{SCME}{extended spatial channel model}
	\acro{SC-FDMA}{single-carrier frequency-division multiple access}
	\acro{SDIV}{spatial diversity}
	\acro{SDMA}{space-division multiple access}
	\acro{SFO}{sampling frequency offset}
	\acro{SFP}{small form-factor pluggable} 
	\acro{SMUX}{spatial multiplexing}
	\acro{SU-MUX}{single user spatial multiplexing}
	\acro{STC}{space-time code}
	\acro{STF}{space-time filter}
	\acro{STVC}{space-time vector coding}
	\acro{SFN}{single frequency network}
	\acro{SF}{shadow fading}
	\acro{SNR}{signal to noise ratio}
	\acro{SIR}{signal to interference ratio}
	\acro{SIC}{successive interference cancellation}
	\acro{SINR}{signal-to-interference-and-noise ratio}
	\acro{SIMO}{single-input multiple-output}
	\acro{SISO}{single-input single-output}
	\acro{SPC}{sum power constraint}
	\acro{SS}{single stream}
	\acro{SSS}{secondary synchronization sequence}
	\acro{ST}{space-time}
	\acro{SW}{stop and wait}
	\acro{SVC}{scalable video coding}
	\acro{SVD}{singular value decomposition}
	\acro{SV}{singular value}
	\acro{TA}{timing advance}
	\acro{TD}{time domain}
	\acro{TDD}{time division duplex}
	\acro{TDMA}{time-division multiple access}
	\acro{TTI}{transmission time interval}
	\acro{Tx}{transmit}
	\acro{TU}{Technical University}
	\acro{TUB}{Technical University of Berlin}
	\acro{TLabs}{Deutsche Telekom Laboratories}
	\acro{TRx}{transceiver}
	\acro{TSVD}{truncated singular value decomposition}
	\acro{TP}{troughput}
	\acro{UE}{user equipment}
	\acro{ULA}{uniform linear array}
	\acro{UDP}{user datagram protocol}
	\acro{PULA}{polarized uniform linear array}
	\acro{UMTS}{Universal Mobile Telecommunications System}
	\acro{UWB}{ultra-wideband}
	\acro{VA}{virtual antenna}
	\acro{V-BLAST}{Vertical Bell Labs Space-Time}
	\acro{VLAN}{virtual local area network}
	\acro{WLAN}{wireless local area network}
	\acro{WCDMA}{wideband code-division multiple access}
	\acro{WDM}{wavelength-division multiplexing}
	\acro{WPAN}{wireless personal area network}
	\acro{ZF}{zero forcing}
	\acro{D2D}{Device to Device}
	\acro{aD2D}{assisted \ac{D2D}}
	\acro{NaD2D}{non-assisted \ac{D2D}}
	\acro{STDMA}{self-organizing time division multiple access}
	\acro{NI}{nominal increment}
	\acro{RR}{report rate}
	\acro{NSS}{nominal start slot}
	\acro{SI}{selection interval}
	\acro{NTS}{nominal transmission slot}
	\acro{NFR}{nominal frequency resource}
	\acro{NS}{nominal slot}
	\acro{STFDMA}{self-organizing time-frequency division multiple access}
	\acro{APSM}{adaptive projected subgradient method}
	\acro{APA}{affine projection algorithm}
	\acro{NLMS}{normalized least mean squares}
	\acro{RKHS}{reproducing kernel Hilbert space}
	\acro{CSMA/CA}{carrier sense multiple access with collision avoidance}
	\acro{V2V}{vehicle to vehicle}
	\acro{V2X}{Vehicle to everything}
	\acro{5G}{fifth generation}
	\acro{C2C}{car to car}
	\acro{VANET}{vehicular ad-hoc network}
	\acro{RRM}{radio resource management}
	\acro{AMC}{adaptive modulation and coding}
	\acro{UL}{uplink}
	\acro{DL}{downlink}    
	\acro{GPS}{global positioning system}
	\acro{SUMO}{Simulation of Urban MObility}
	\acro{RBIR}{received bit information rate}
	\acro{KPI}{key performance indicator}
	\acro{ADMM}{alternating direction method of multipliers}
	\acro{SLF}{spatial loss field}
	\acro{RBF}{radial basis function}
	\acro{MTC}{machine-type communications}
	\acro{NMSE}{normalized mean squared error}
	\acro{A2A}{Any-to-any}
	\acro{TPT}{tomographic projection technique}
	\acro{gAM}{generalized alternating minimization}
	\acro{GEMV2}{Geometry-based Efficient propagation Model for V2V communication}
\end{acronym}

\bibliographystyle{IEEEtran}
\bibliography{referencesICC}
\addcontentsline{toc}{part}{Bibliography}


\end{document}